\documentclass[11pt,a4paper]{article}
\usepackage[authoryear]{natbib}
\usepackage{graphics}
\usepackage{ifpdf}
\usepackage{graphicx}
\usepackage{amssymb,amstext}
\usepackage{amsmath,latexsym,xspace}
\usepackage{array}
\usepackage{setspace} 
\usepackage{comment}
\usepackage{authblk}
\usepackage{diagbox}
\usepackage{multirow}
\usepackage{caption}
\usepackage{subcaption}
\usepackage{hyperref}

\usepackage{amsthm}

\usepackage[scale={0.7,0.85},centering,includeheadfoot]{geometry}

\newtheorem{theorem}{Theorem}[section]

\def\E{\text{E}}
\def\Cov{\text{Cov}}
\def\mm#1{\ensuremath{\boldsymbol{#1}}} 
\def\AR1{AR$(1)$\xspace}

\setkeys{Gin}{width=0.45\textwidth}

\begin{document}
\begin{titlepage}
    \title{An approximate fractional Gaussian noise model with
        ${\mathcal O}(n)$ computational cost}
    \author[1]{ Sigrunn H.\ S\o rbye$^{1}$, Eirik Myrvoll-Nilsen}
    \author[2]{H\aa vard Rue}
    \affil[1]{Department of Mathematical Sciences, UiT The Arctic
        University of Norway, Troms{\o}, Norway}
    \affil[2]{CEMSE Division, King Abdullah University of Science and
        Technology, Thuwal, Saudi Arabia} \date{\today}
\end{titlepage}
\maketitle

\begin{abstract}
    Fractional Gaussian noise (fGn) is a stationary time series model
    with long memory properties applied in various fields like
    econometrics, hydrology and climatology. The computational cost in
    fitting an fGn model of length $n$ using a likelihood-based approach
    is ${\mathcal O}(n^{2})$, exploiting the Toeplitz structure of the
    covariance matrix. In most realistic cases, we do not observe the
    fGn process directly but only through indirect Gaussian observations, so
    the Toeplitz structure is easily lost and the computational cost
    increases to ${\mathcal O}(n^{3})$. This paper presents an approximate fGn model of 
    ${\mathcal O}(n)$ computational cost,  both with direct or indirect
    Gaussian observations, with or without conditioning. This is
    achieved by approximating fGn with a weighted sum of independent
    first-order autoregressive processes, fitting the parameters of
    the approximation to match the autocorrelation function of the fGn
    model. The resulting approximation is stationary despite being
    Markov and gives a remarkably accurate fit using only four
    components. The performance of the
    approximate fGn model is demonstrated in simulations and two real data examples.  
\end{abstract}

{\bf Keywords:} Autoregressive processes, Gaussian Markov random field, long memory, R-INLA, Toeplitz matrices
 
\section{Introduction}

Many natural processes observed in either time or space exhibit long
memory dependency structures. One way to characterise long memory is
in terms of the autocorrelation function having a slower than
exponential decay, as a function of increasing temporal or geographical distance 
between observational points. In 
second-order stationary time series, long memory implies that the 
autocorrelations are not absolutely summable \citep{mcleod:78}.
Long memory behaviour has been observed within a vast variety of time
series applications, like hydrology \citep{hurst:51, hosking:84},
geophysical time series \citep{mandelbrot:69}, network traffic
modelling \citep{willinger:96}, econometrics \citep{
    baillie:96, cont:05} and climate data analysis \citep{franzke:12, rypdal:14}.
For comprehensive introductions to long memory processes and their
applications, see for example  \cite{taqqu:03} and \cite{beran:13}.

Fractional Gaussian noise is defined as the increment process of
fractional Brownian motion. It is a stationary and parsimoniously
parameterised model, with dependency structure characterised by the
Hurst exponent $H$, 
which gives long memory when $1/2 < H < 1$. 
The computational cost of likelihood-based inference in fitting an fGn process 
of length $n$ is ${\mathcal O}(n^{2})$, using the Durbin-Levinson or Trench algorithms \citep{mcleod:07,golub:96, durbin1960, levinson1947, trench:64}. These algorithms make use of the Toeplitz structure of the covariance matrix of fGn and are considered to be sufficiently fast when $n$ is not too large  \citep{mcleod:07}.  A main 
problem  is that the required Toeplitz structure is fragile to modifications of the Gaussian 
observational model and  computations
of conditional distributions. For example, the Toeplitz structure is destroyed if fGn is observed indirectly 
with Gaussian inhomogeneous noise, or has missing data. Without the Toeplitz structure, 
the computational cost in fitting fGn increases to
${\mathcal O}(n^{3})$, which is far too expensive in many real data applications. 

This paper presents an accurate and computationally efficient 
approximate fGn model of cost  ${\mathcal O}(n)$, both
with direct and indirect Gaussian observations, with or without
additional conditioning. This represents a major computational
improvement compared with existing approaches, and allows routinely use
of fGn models with large $n$, with negligible loss of accuracy. The new approximate model uses
 a weighted sum of independent first-order autoregressive processes (\AR1).
 The motivation is that aggregation of short-memory processes plays an important role to explain long memory behaviour in time series  \citep{beran:10} and an infinite weighted sum of AR(1) processes will give long memory 
\citep{granger1980}. However, the physical intuition about long memory processes suggests that we leverage this
relationship much more than the asymptotic result states, see for
example \cite{aggregationsims}.  The new approximate model only needs a weighted sum of a few AR(1) 
processes to be accurate. We obtain this by fitting the weights and the coefficients of the approximation
to mimic the autocorrelation function of the exact fGn model, as a
continuous function of $H$. 

A key feature of the approximate fGn model is a high degree of
conditional independence within the model. Specifically, the
approximate model will be represented as a Gaussian Markov random
field (GMRF), which computational properties are not destroyed by
indirect Gaussian observations nor conditioning~\citep{ruebok}. The
approximate model is also stationary, a desired property which is not
common among GMRFs, as they typically have boundary effects. Since the
approximate model is a local GMRF, it also fits well within the framework
of latent Gaussian models for which approximate Bayesian
analysis is obtained with integrated nested Laplace approximations (INLA)
\citep{inlartikkel} using the \texttt{R}-package \texttt{R-INLA} (\href{http://www.r-inla.org}{\texttt{http://www.r-inla.org}}).  A different aspect is that an aggregated model of a few \AR1 
components could actually represent a more plausible and interpretable model than fGn in
real data applications. Specifically, the approximate model can serve as a tool for automatic source separation in
situations where the data at hand represent combined signals.

The plan of this paper is as follows. Section~\ref{sec:fgn} reviews
the computational cost in fitting the exact fGn model to direct and
indirect Gaussian observations. Section~\ref{sec:approx} presents the
new approximate fGn model and derives the computational cost and
memory requirement for evaluating the log-likelihood. The performance
of the approximate model is demonstrated by simulations in
Section~\ref{sec:results}, also including a study of its predictive
properties. In Section~\ref{sec:real-data}, we use the implicit source separation ability in decomposing the historical 
dataset of annual water level minima for the Nile river \citep{toussoun:1925, beranbok}. Implementation within the class of 
latent Gaussian models is demonstrated in analysing a
monthly mean surface air temperature series for Central England \citep{manley:53,hadcet,parker:92}. 
 Concluding remarks are given in
Section~\ref{sec:conclusion}.

\section{The computational cost of fGn}
\label{sec:fgn}

Let $\mm{x}=(x_1,\ldots , x_n)^T$ be a zero-mean fGn process, $\mm{x}\sim {\mathcal N}_n(0,\mm{\Sigma})$. The elements of the
covariance matrix, $\Sigma_{ij}=\sigma^2\gamma_{\mm{x}} (k)$, $k=|i-j|$, are
defined by the autocorrelation function
\begin{equation*}
    \gamma_{\mm{x}} (k) = \frac{1}{2} \Big(  |k-1|^{2H} -2|k|^{2H} +
    |k+1|^{2H}  \Big),\quad k=0,1,\ldots , n-1.
    \label{eq:acf-fgn}
\end{equation*}
This function has a hyperbolic decay
$\gamma_{\mm{x}}(k)\sim H(2H-1)k^{2(H-1)}$ as $k\rightarrow \infty$. The fGn process has long memory when $1/2 < H < 1$, while 
it  reduces to white noise when $H=1/2$. When  $0< H < 1/2$, the fGn model has anti-persistent properties, but
we do not discuss this case here.

When fGn is observed directly, we  estimate $H$ by maximizing the log-likelihood function
\begin{equation*}
    \log(\pi(\mm{x}))=-\frac{n}{2}\log(2\pi) +\frac{1}{2}\log|\mm{Q}| 
    -\frac{1}{2} \mm{x}^T\mm{Q}\mm{x},\label{eq:mle}
\end{equation*}
where $\mm{Q}=\mm{\Sigma}^{-1}$ is the precision matrix of $\mm{x}$.
Making use of the Toeplitz structure of
$\mm{\Sigma}$, the likelihood is evaluated in $\mathcal{O}(n^2)$
flops using the Durbin-Levinson algorithm \citep[Algorithm 4.7.2]{golub:96}.
Also, the precision matrix $\mm{Q}$ can be calculated in
$\mathcal{O}(n^2)$ flops by the Trench algorithm \citep[Algorithm
4.7.3]{golub:96}. In general, the Trench algorithm can be combined
with the Durbin-Levinson recursions \citep[Algorithm 5.7.1]{golub:96},
to calculate the exact likelihood of general linear Gaussian time
series models \citep{mcleod:07}.

A major drawback of relying on these algorithms for Toeplitz matrices
is that the Toeplitz structure is easily destroyed if the time series is observed indirectly. 
Consider a simple regression setting in
which an fGn process is observed with independent Gaussian random
noise, 
\begin{equation}
    \mm{y}=\mm{x}+\mm{\epsilon},\label{eq:regr}
\end{equation}
where $\mm{\epsilon}\sim {\mathcal N}_n(\mm{0},\mm{D}^{-1})$ and $\mm{D}$ is
diagonal. The log-density of $\mm{x}\mid\mm{y}$ is
\begin{equation*}
    \log \pi(\mm{x}\mid \mm{y}) = 
    \frac{1}{2}\log|\mm{Q} + \mm{D}| 
    -\frac{1}{2}\mm{x}^T(\mm{Q}+\mm{D})\mm{x}
    +\mm{y}^T\mm{D}\mm{x} + \text{constant}.
    \label{eq:conditioning}
\end{equation*}
The conditional mean of $\mm{x}$ is found by solving
$(\mm{Q}+\mm{D})\mm{\mu}_{\mm{x}\mid \mm{y}}=\mm{D}\mm{y}$ with
respect to $\mm{\mu}_{\mm{x}\mid \mm{y}}$, while the marginal
variances equal $\mbox{diag}\{(\mm{Q}+\mm{D})^{-1}\}$. The Toeplitz
structure of $\Cov(\mm{y})=\mm{Q}^{-1}+\mm{D}^{-1}$ is only retained when the noise term has
homogeneous variance, i.e.\ $\mm{D}^{-1} \propto\mm{I}$. With
non-homogenous observation variance or missing data, 
the computational cost in fitting \eqref{eq:regr} would
require general algorithms of cost $\mathcal{O}(n^3)$. This
makes analysis of many real data sets infeasible, or at best challenging.

The motivation for expressing the log-likelihood function in terms of
the precision matrix $\mm{Q}$, is to prepare for an approximate GMRF
representation of the fGn model. We have already noted that
aggregation of an infinite number of short-memory processes can explain
long memory behaviour in time series. This implies
that $\mm{Q}$ is (or can be approximated with) a sparse band matrix,
but with a larger dimension (still denoted by $n$) for a finite sum.
Assume for a moment that such an approximation exists.  We can then apply 
general numerical algorithms for sparse matrices which only depend on the non-zero
structure of the matrix. This implies that the numerical cost in 
dealing with $\mm{Q}$ or $\mm{Q} + \mm{D}$, is the same. Conditioning
on subsets of $\mm{x}$ implies nothing else than working with a submatrix
of $\mm{Q}$ or $\mm{Q}+\mm{D}$, and does not add to the computational
costs; see \cite[Ch.~2]{ruebok} for details. Specifically, we can
make use of the Cholesky decomposition, in which the relevant precision matrix
$\mm{Q}+\mm{D}$ is factorised as $\mm{Q} + \mm{D}=\mm{L}\mm{L}^T$,
where $\mm{L}$ is a lower triangular matrix. The log-likelihood
is then evaluated with negligible cost  as the
log-determinant is $\log|\mm{Q} + \mm{D}|=2\sum_{i=1}^n L_{ii}$ \citep{rue:01}. The
conditional mean is found by solving $\mm{L}\mm{u} = \mm{D}\mm{y}$ and
$\mm{L}^{T}\mm{\mu}_{\mm{x}\mid \mm{y}} = \mm{u}$. The numerical cost
in finding the Cholesky decomposition depends on the non-zero
structure of the matrix. For time-series models (or long skinny
graphs), the cost is ${\mathcal O}(n)$ \citep{ruebok}. 
The explicit construction of such an approximation is discussed next.

\section{An approximate fGn model}\label{sec:approx}
This section presents an approximate fGn
model which is a weighted sum of just a few independent \AR1 processes. 
We will fit the parameters of the approximation 
 to mimic the autocorrelation structure of fGn
up to a given finite lag. The resulting approximate model is a GMRF 
with a banded precision matrix of fixed
bandwidth, which gives a computational cost of  ${\mathcal O}(n)$.

\subsection{Fitting the autocorrelation function}

Define $m$ independent \AR1 processes by
\begin{equation}
    z_{j,t}= \phi_{j}z_{j,t-1} + \nu_{j,t}, \qquad j=1,\ldots , m,
    \quad t=1,\ldots , n,\label{eq:ar1}
\end{equation}
where $0< \phi_j <1$ denotes the first-lag autocorrelation
coefficient of the $j$th \AR1 process. Further, let
$\{\nu_{j,t}\}_{j=1}^m$ be independent zero-mean Gaussians, with
variance $\sigma^2_{\nu,j}=(1-\phi_j^2)$. Define the
cross-sectional aggregation of the \AR1 processes,
\begin{equation}
    \tilde{\mm{x}}_m= \sigma \sum_{j=1}^m\sqrt{w_j} \mm{z}^{(j)},
    \label{eq:optimgrunnlag0}
\end{equation}
where $\mm{z}^{(j)} = (z_{j,1},\ldots ,z_{j,n})^T$ and
$\sum_{j=1}^m w_j=1$. This implies that $\mbox{Var}( \tilde{\mm{x}}_m)=\sigma^2$. The finite-sample properties of a 
similar aggregation of AR(1) processes are studied in \cite{aggregationsims}, where $w_j=1/m$, $\sigma^2_{\nu}=1$ and 
where the coefficients $\phi_j$ are Beta distributed. 
They conclude that 
``\emph{one should
be aware that cross-sectional aggregation leading to long memory is an asymptotic feature
 that applies for the cross-sectional dimension tending to infinity. In finite samples and for moderate cross-sectional dimensions
 the observed memory of the series can be rather different from the theoretical memory}''. 
 
The approximation presented here only needs a small value of the cross-sectional dimension $m$
to be accurate.  The key idea to our approach is to fit the weights $\mm{w}=\{w_j\}_{j=1}^m$ and the
autocorrelation coefficients $\mm{\phi}=\{\phi_j\}_{j=1}^m$ in \eqref{eq:optimgrunnlag0}
to match the autocorrelation function of fGn, as a function of $H$. 
The autocorrelation function of
\eqref{eq:optimgrunnlag0} follows directly as
\begin{equation*}
    \gamma_{\tilde{\mm{x}}_m}(k) = \sum_{j=1}^m w_j \phi_j^{|k|}, 
    \quad k=0,1,\ldots , n-1.
    \label{eq:acf-approx}
\end{equation*} 
Now, fix a value of $1/2<H<1$. We fit the  weights and coefficients
$(\mm{w},\mm{\phi})_H$ by minimizing the weighted squared
error
\begin{equation}
    (\mm{w},\boldsymbol \phi)_H = 
    \underset{(\mm{w},\boldsymbol \phi)}{\text{argmin}} 
    \sum_{k=1}^{k_{\max}} \frac{1}{k}\big(\gamma_{\tilde{\mm{x}}_m}(k)
    -\gamma_{\mm{x}}(k)  \big)^2, 
    \label{eq3:optimprosedyre}
\end{equation}
where $k_{\max}$ represents a user-specified upper limit (we use
$k_{\max} = 1000$). The squared error is weighted by $1/k$  to ensure a good fit for
the autocorrelation function close to lag 0, while less weight is given
to tail behaviour as the autocorrelation function is decaying slowly.

By a quite huge calculation done only once, we find
$(\mm{w},\boldsymbol \phi)_H$ for a fine grid of $H$-values. 
Spline interpolation is used for values of $H$ in between, to represent the weights and coefficients as
continuous functions of $H$. The interpolation and 
fitting are performed using reparameterised weights and coefficients to ensure
uniqueness and improved numerical behaviour. These reparameterisations
are defined as
\begin{equation*}
    w_j = \frac{e^{v_j} }{\sum_{i=1}^m e^{v_i}} \quad\text{and}\quad
    \phi_j = \frac{1}{1 + \sum_{i=1}^j e^{-u_i}},
    \label{eq3:weightparam}
\end{equation*}
where $j=1,\ldots , m$ and 
where $v_1=0$. The Hurst exponent is transformed as
$H = 1/2 + 1/2 \exp(h)/(1+\exp(h))$. This ensures a stable and
unconstrained parameter space on $\mathbb{R}^{2m-1}$ for fixed $h$,
where $\phi_1 > \cdots > \phi_m$. Note that the error of the fit tends
to zero, when $H$ goes to 1 or $1/2$. The resulting coefficients and
weights for $m=3$ and $m=4$ are displayed in Figure~\ref{fig:mapping}.
The fitted weights and coefficients are also available in \texttt{R} using the
function \texttt{INLA::inla.fgn}.

\begin{figure}[ht]
    \begin{center}
    \includegraphics[width=0.4\textwidth,angle=0]{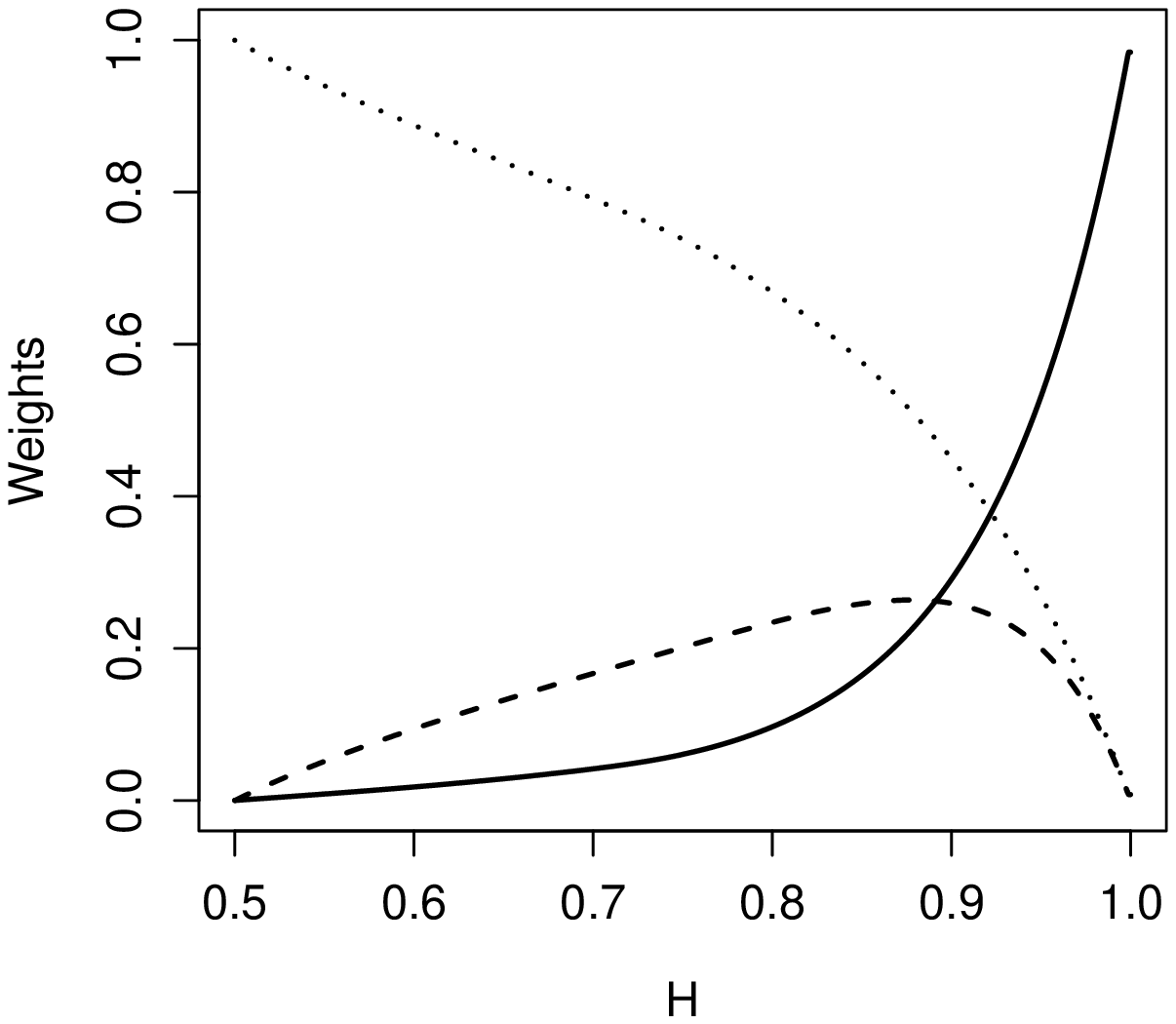}\hspace{0.5cm}
    \includegraphics[width=0.4\textwidth,angle=0]{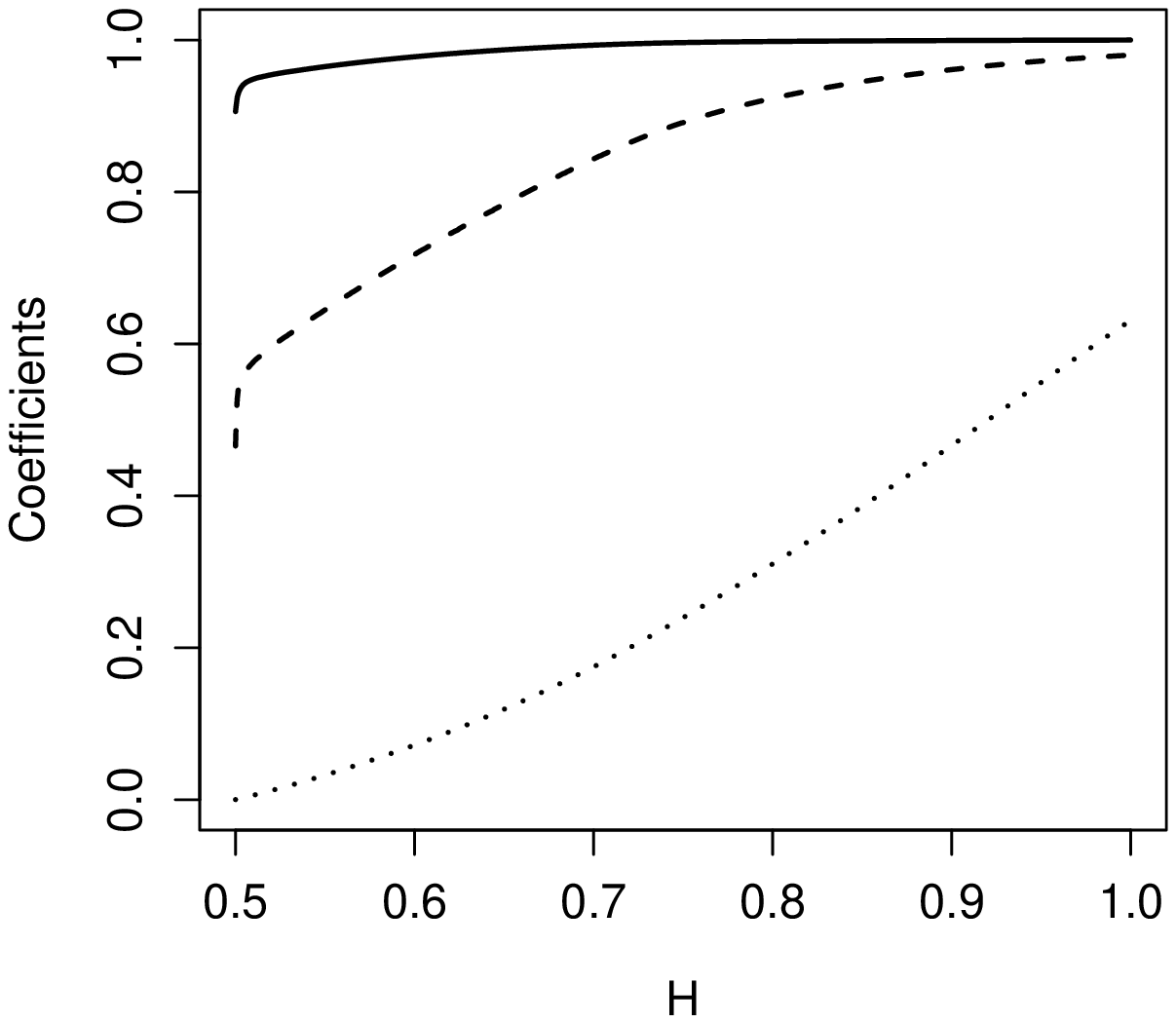}\\
    \includegraphics[width=0.4\textwidth,angle=0]{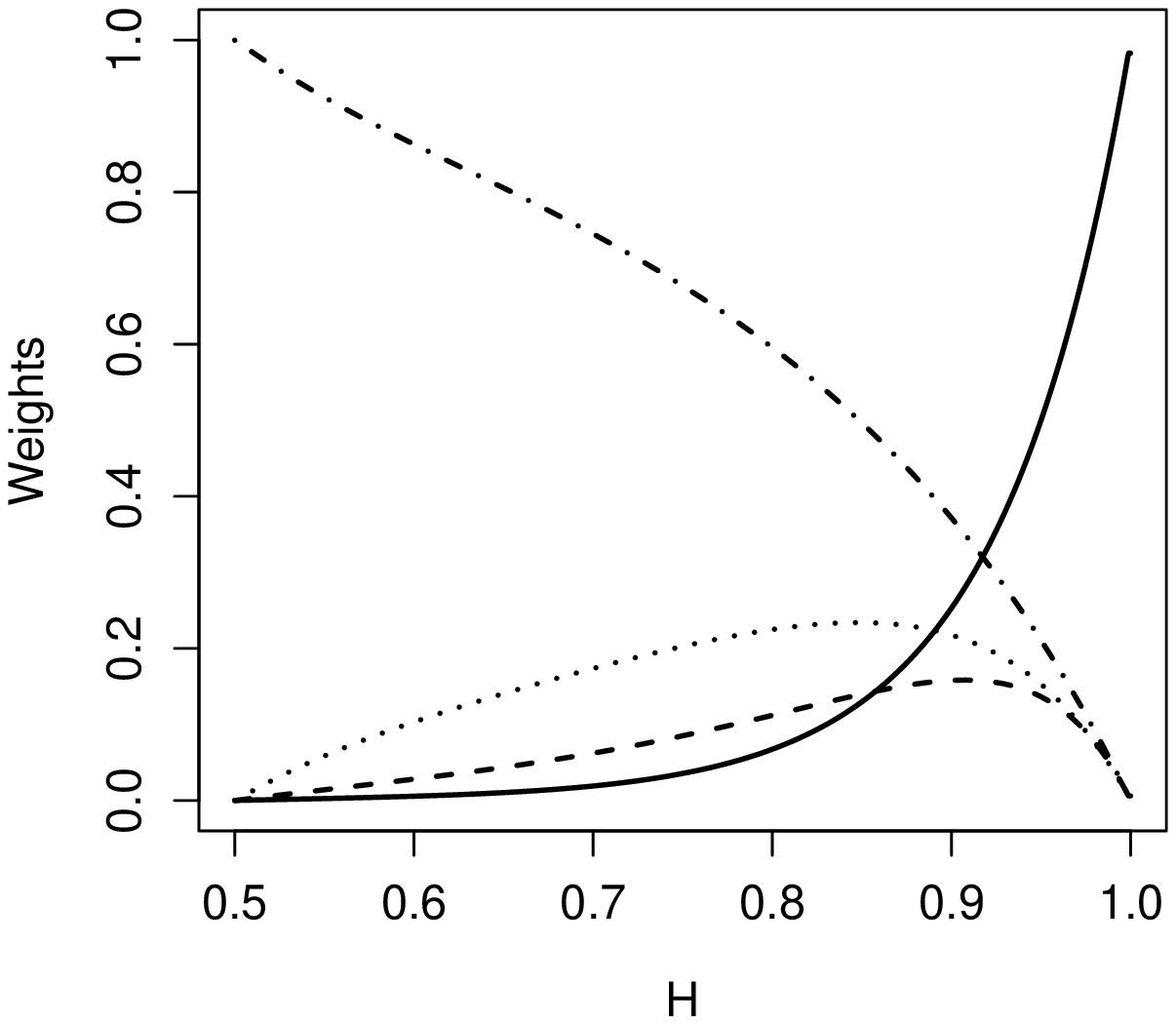}\hspace{0.5cm}
    \includegraphics[width=0.4\textwidth,angle=0]{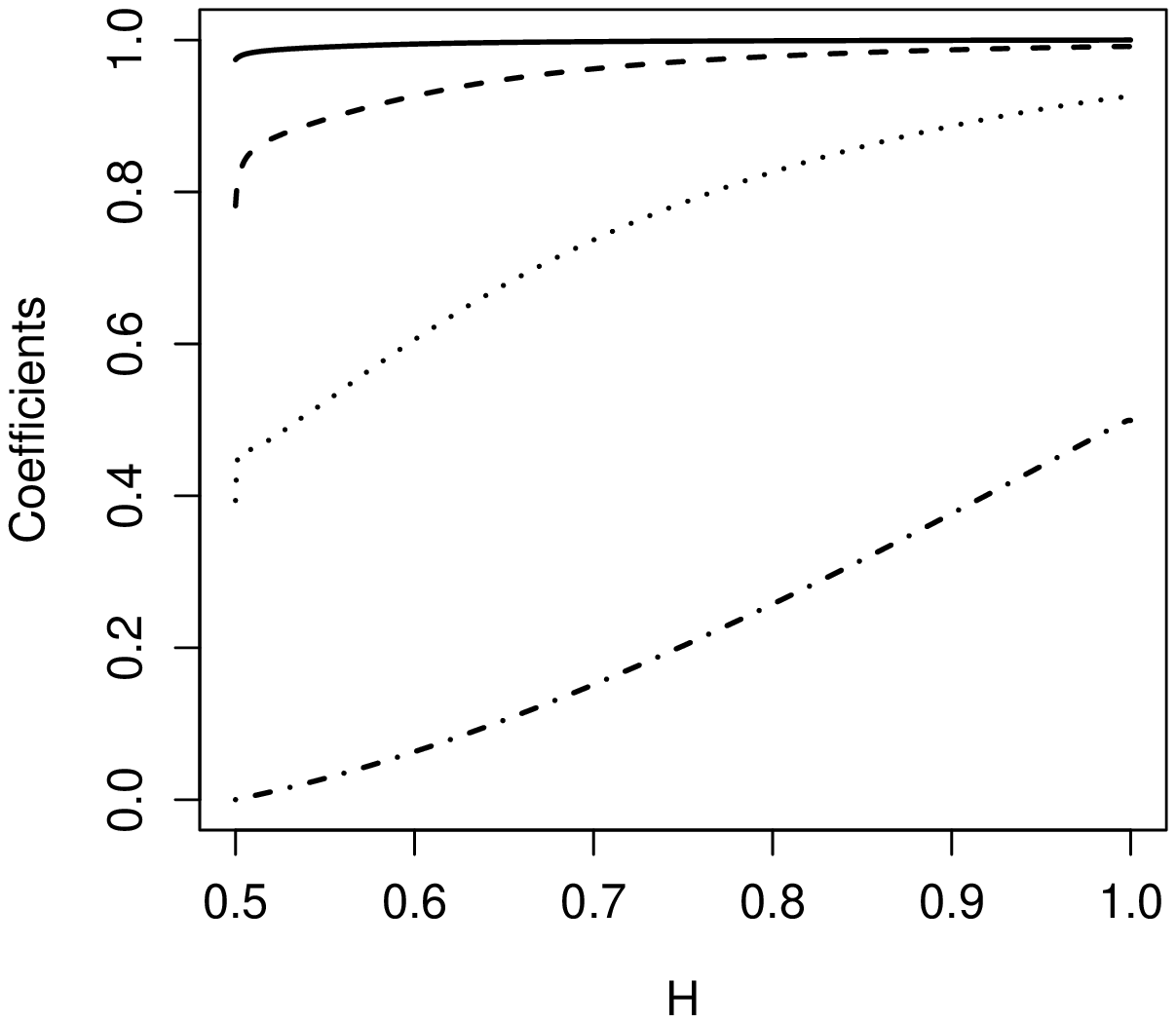}
     \caption{Upper panels: The fitted weights and coefficients
        $(\mm{w},\boldsymbol \phi)_H$ in \eqref{eq3:optimprosedyre} 
        using $m=3$ components in the approximation. Lower panels: Similar using $m=4$.} 
    \label{fig:mapping}
    \end{center}
\end{figure}

\subsection{A Gaussian Markov random field representation}
\label{sec:cost}

We will now discuss the precision matrix for the approximate fGn
model. We start with one \AR1 process~\eqref{eq:ar1} of length $n$, with unit variance
and a tridiagonal precision matrix
\begin{displaymath}
    \mm{R}(\phi_j)  = \frac{1}{1-\phi_j^{2}}
    \begin{pmatrix}
        1 & -\phi_j&& &\\
        -\phi_j & 1+\phi_j^2 & -\phi_j &&\\
        &\ddots&\ddots & \ddots & &\\
        & &-\phi_j & 1+\phi_j^2 & -\phi_j\\
        &&& -\phi_j & 1
    \end{pmatrix}.
\end{displaymath}        
For the approximate
fGn model, we have $m$ such processes and their sum. Hence we need the
$(m+1)n\times (m+1)n$ precision matrix of the vector
\begin{equation}
(\tilde{\mm{x}}_m^T, \mm{z}^{(1)T},\ldots , \mm{z}^{(m)T}).\label{eq:MGMRF}
\end{equation}
To ensure a non-singular distribution, we will add a small Gaussian
noise term to the sum, i.e.\ we let
\begin{equation}
    \tilde{\mm{x}}_m=\sigma \left(\sum_{j=1}^m\sqrt{w_j} \mm{z}^{(j)}+\mm{\epsilon}\right),  
    \label{eq:optimgrunnlag1}
\end{equation}
where the precision of $\mm{\epsilon}$ is high, like $\kappa=\exp(15)$.
The (upper part of the) precision matrix is found as
\begin{equation*}
    \begin{pmatrix}
        \begin{tabular}{rrrrr}
          $ \kappa\mm{I}/\sigma^{2} $& $-\sqrt{w_1}  \kappa \mm{I}/\sigma $ & 
                                                                    $-\sqrt{w_2}  \kappa\mm{I}/\sigma $& $\ldots $&  $ -\sqrt{w_m}  \kappa\mm{I}/\sigma $\\
                           & $\mm{R}(\phi_1) + w_1 \kappa \mm{I} $  &$ \sqrt{w_1 w_2} \kappa \mm{I} $ & $\ldots$  & $ \sqrt{w_1 w_m} \kappa \mm{I}$ \\
                           &     & $ \mm{R}(\phi_2) + w_2 \kappa \mm{I} $ &    $\ddots  $  & $ \vdots $ \\
                           & & & $\ddots$ & $\sqrt{w_{m-1} w_m} \kappa \mm{I} $\\
                           &  & &  & $ \mm{R}(\phi_m) +w_m\kappa\mm{I} $\\
        \end{tabular}
    \end{pmatrix}.
    \label{eq:precision}
\end{equation*}
The non-zero structure is displayed in Figure~\ref{fig:matrices} (left
panel) for $m=3$ and $n=10$. Even though the matrix is sparse, a more
optimal structure can be achieved by grouping the $m+1$ variables
associated with each of the $n$ time points,
\begin{equation}
    \mm{v}=\left(\tilde{x}_{m1},z^{(1)}_1,\ldots ,z^{(m)}_1,\;
      \tilde{x}_{m2},z^{(1)}_2,\ldots ,z^{(m)}_2,\;\ldots, \;
      \tilde{x}_{mn},z^{(1)}_n,z^{(2)}_n,\ldots ,  z^{(m)}_n\right)^T.
    \label{eq:sparse}
\end{equation}
The benefit of this reordering is that the corresponding precision
matrix $\mm{Q}_v$ is a band matrix, see Figure~\ref{fig:matrices}
(middle panel). Doing the Cholesky decomposition,
$\mm{Q}_v=\mm{L}_v\mm{L}_v^T$, the lower triangular matrix
$\mm{L}_v$ will inherit the lower bandwidth of $\mm{Q}_v$ \citep[Thm.~4.3.1]{rue:01,golub:96}, see Figure~\ref{fig:matrices} (right panel). This
leads to the following key result concerning the computational cost of
the approximate model, with a trivial proof.
\begin{theorem}\label{teo:flops}
    The number of flops needed for Cholesky decomposition of $\mm{Q}_v$
    is $n(m+1)^3$. The memory requirement for the Cholesky triangle
    is $n(m+1)(m+2)$ reals.
\end{theorem}
\begin{proof}
    $\mm{Q}_v$ is a band matrix with dimension $d=n(m+1)$ and
    bandwidth $b=m+1$. The computational cost of the Cholesky
    factorisation, $\mm{Q}_v=\mm{L}_v\mm{L}_v^T$ is $db^2=n(m+1)^3$
    and the memory requirement needed is $d(b+1)=n(m+1)(m+2)$
    \citep[Section~4.3.5]{golub:96}.
\end{proof}

\begin{figure}
    \begin{center}
        \includegraphics[width=0.3\textwidth]{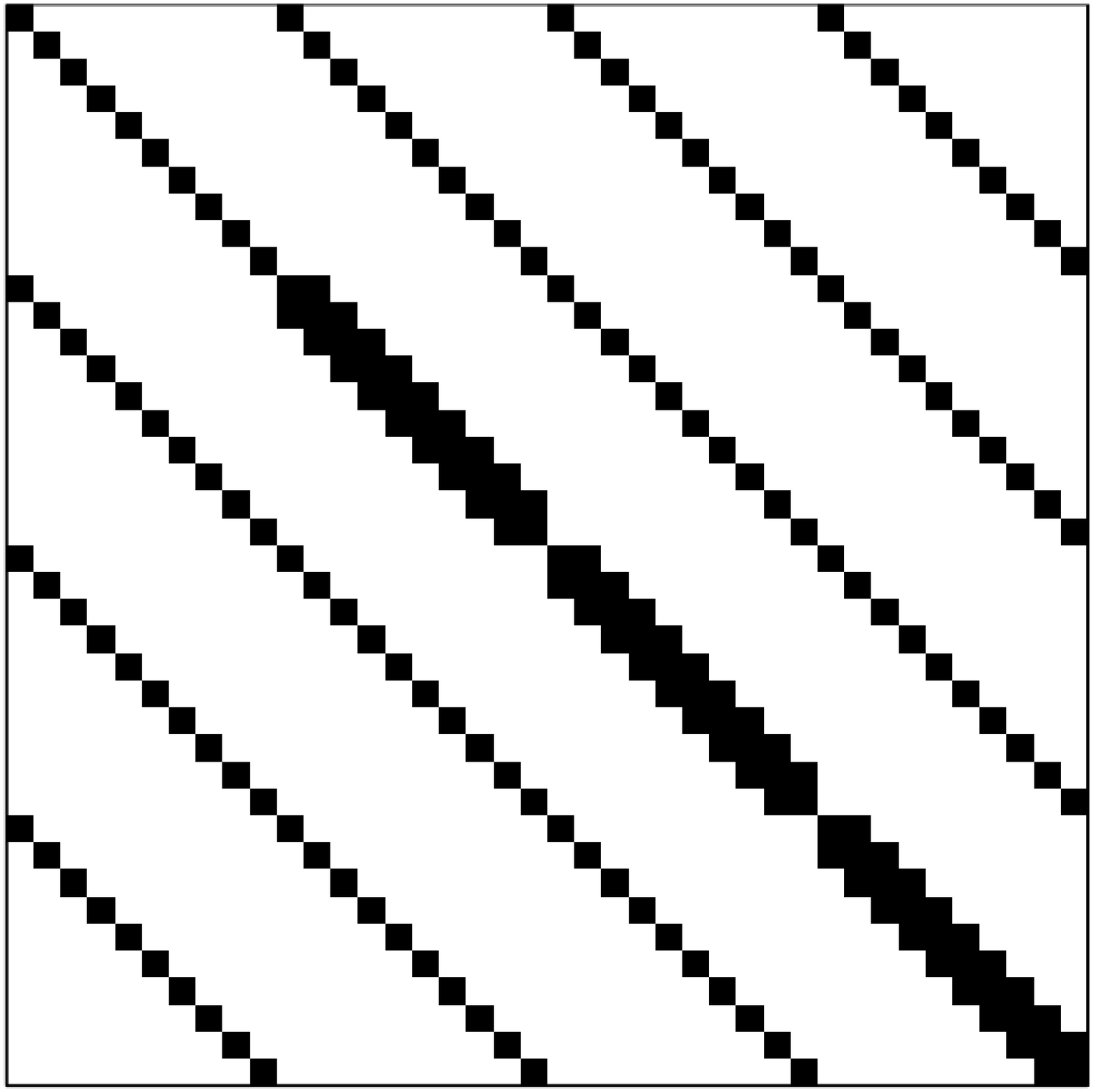}\hspace{0.5cm}
        \includegraphics[width=0.3\textwidth]{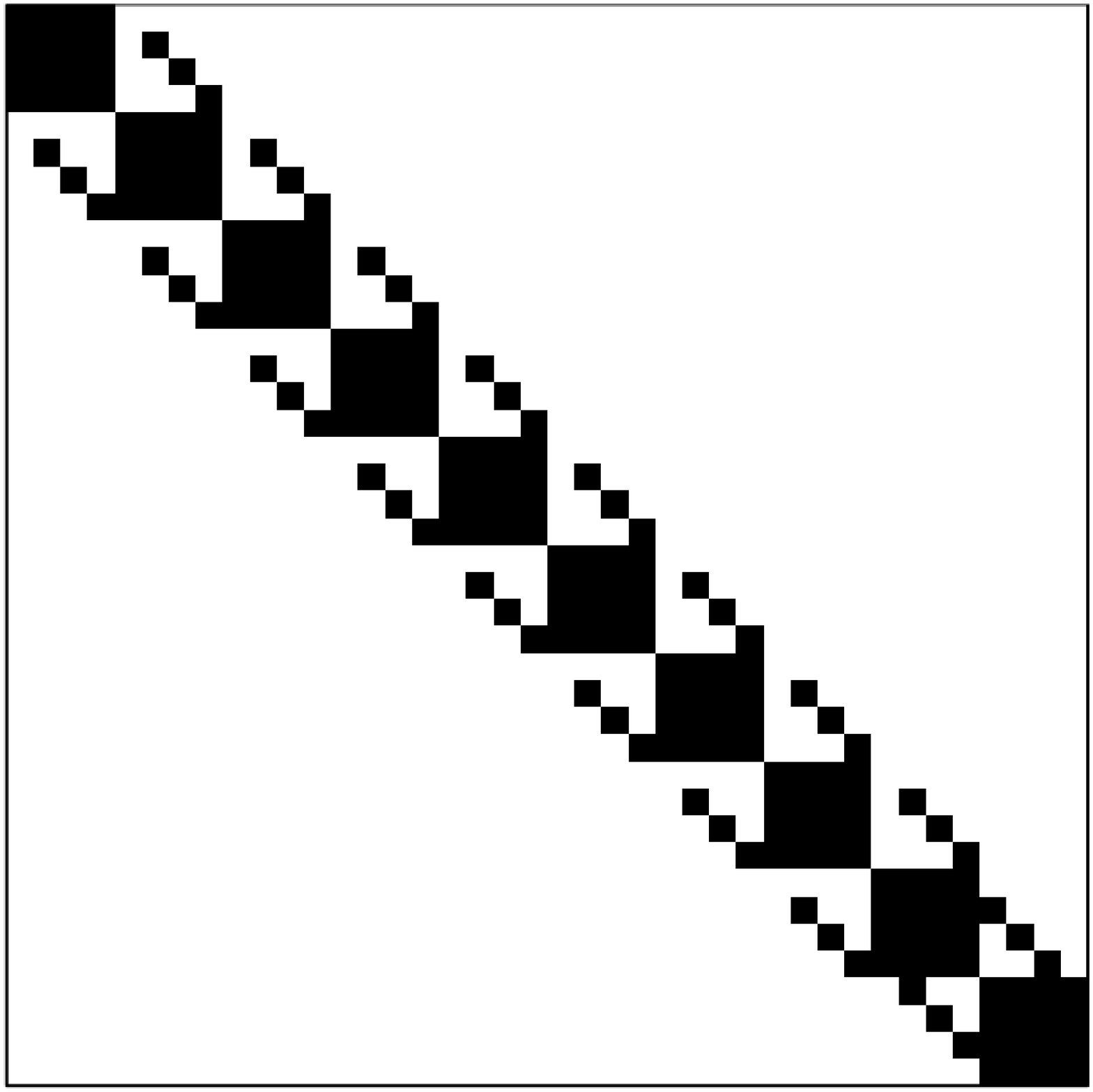}\hspace{0.5cm}
        \includegraphics[width=0.3\textwidth]{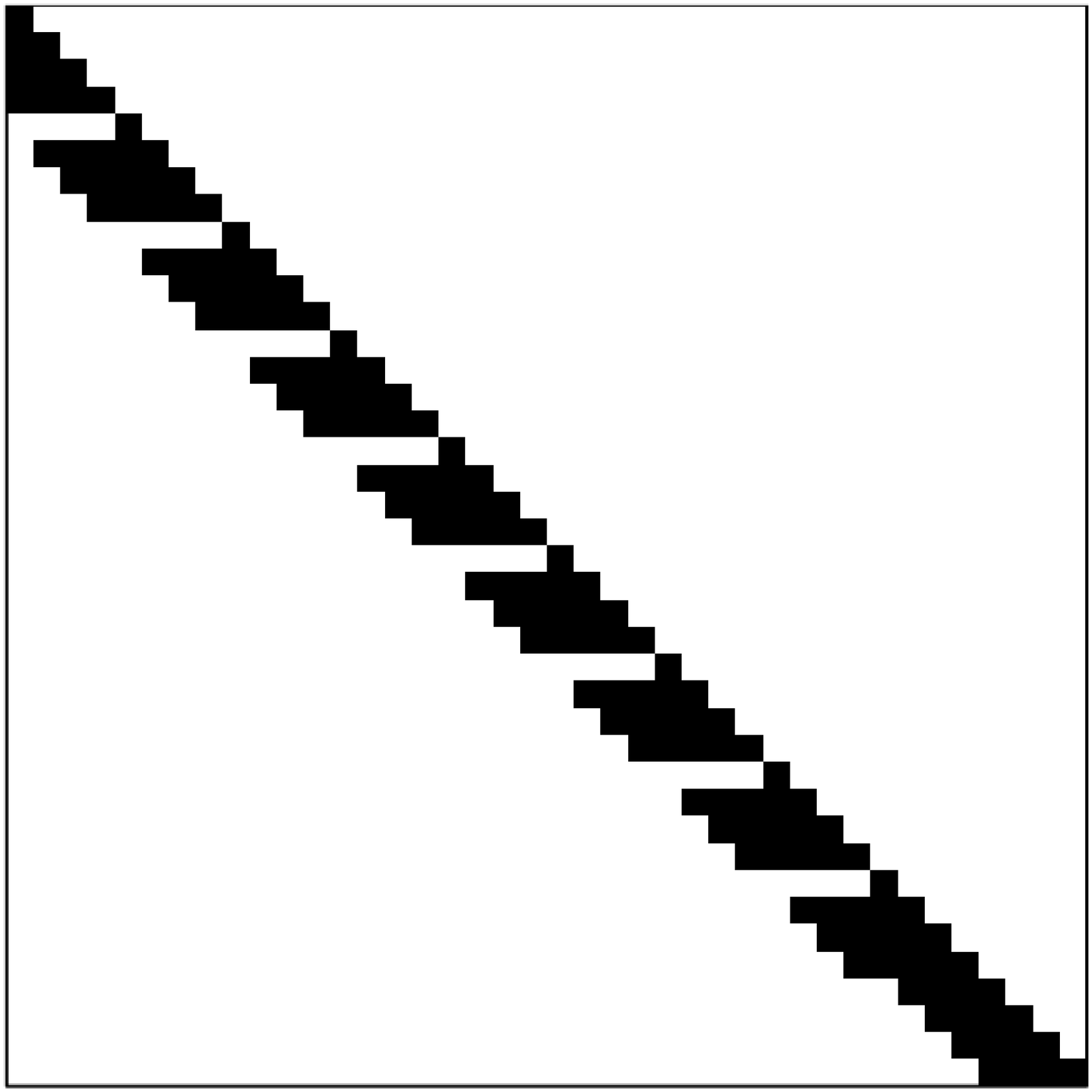}
        \caption{Left panel: The structure of the precision matrix of the vector  
            \eqref{eq:MGMRF}. Middle panel: The structure of the precision matrix of the reordered vector \eqref{eq:sparse}. Right panel: The
            resulting structure of the lower triangular matrix in the
            Cholesky decomposition. The matrices are illustrated for the case 
            $m=3$  and $n=10$.}
        \label{fig:matrices}
    \end{center}
\end{figure}

The computational cost and memory requirement of the Cholesky
decomposition do not change if the approximate fGn model is observed
indirectly, like in the regression model \eqref{eq:regr}. 
Notice that it is also possible to construct an approximation using the cumulative sums of
$\sigma\sum_{j=1}^m\sqrt{w_j} \mm{z}^{(j)}$ to form a sparse $mn\times mn$
precision matrix, with the same bandwidth.  Obviously, this approach 
gives computational 
savings but it does not 
allow for automatic source separation in situations where the fGn can be seen to represent 
combined signals. This feature of the approximate model is demonstrated in Section~\ref{sec:nile}.

\subsection{Choosing the number of AR(1) components in the approximation}
\label{sec:accuracy}

The choice of $m$ in \eqref{eq:optimgrunnlag1} reflects a trade-off
between computational efficiency and approximation error. This implies that 
$m$ should be as small as possible but still large enough to give a reasonably accurate approximation of 
the autocorrelation function of fGn.  Figure~\ref{fig:acf} illustrates the autocorrelation function
of fGn compared with the approximate model when $m=3$ and $m=4$, using $k_{\max}=1000$ in \eqref{eq3:optimprosedyre}. 
We only show results for $H=0.9$ as the differences between the curves will be less visible 
using smaller values of $H$. 

\begin{figure}[ht]
    \begin{center}
    \includegraphics[width=0.4\textwidth]{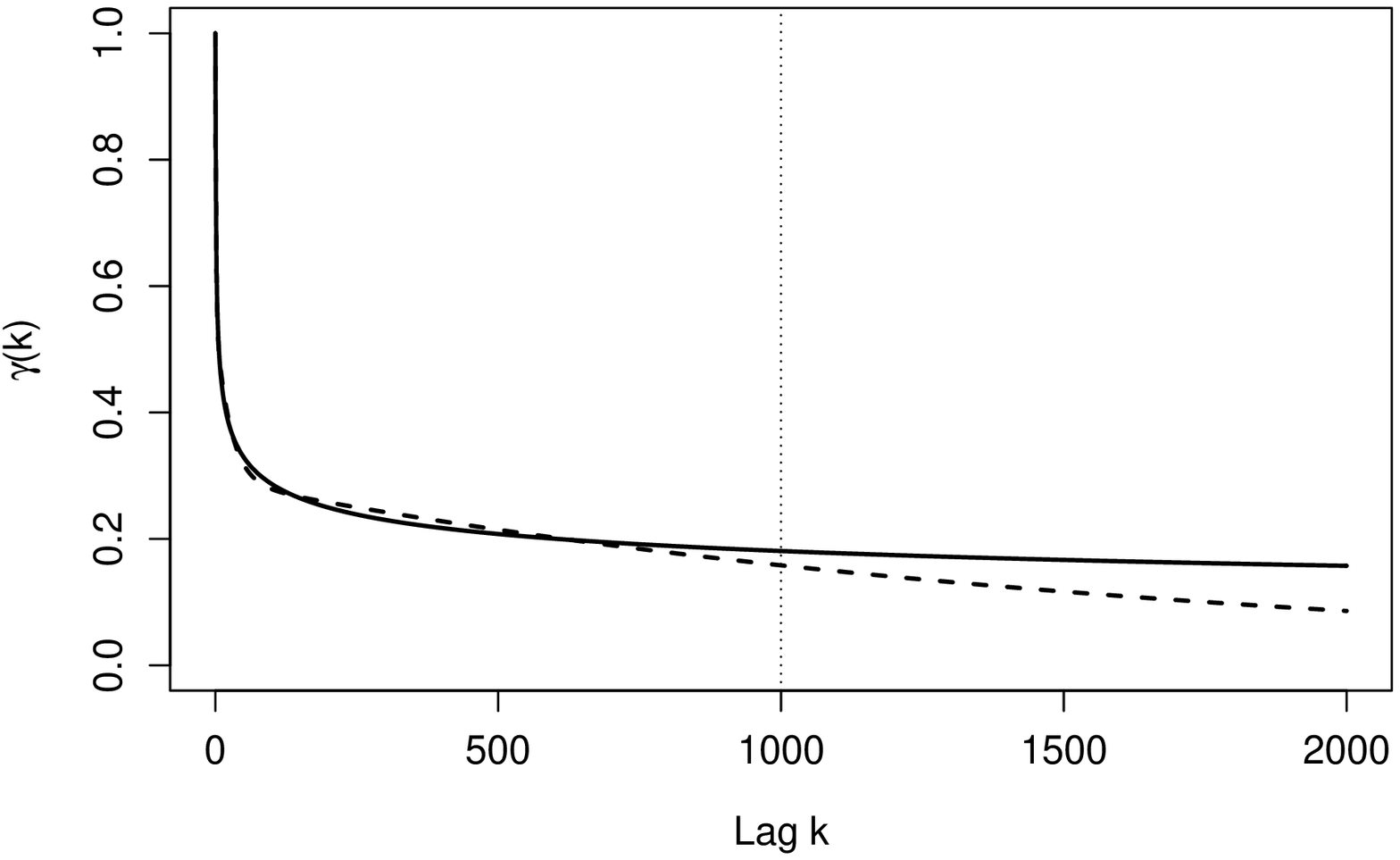}
    \hspace{0.5cm}
    \includegraphics[width=0.4\textwidth]{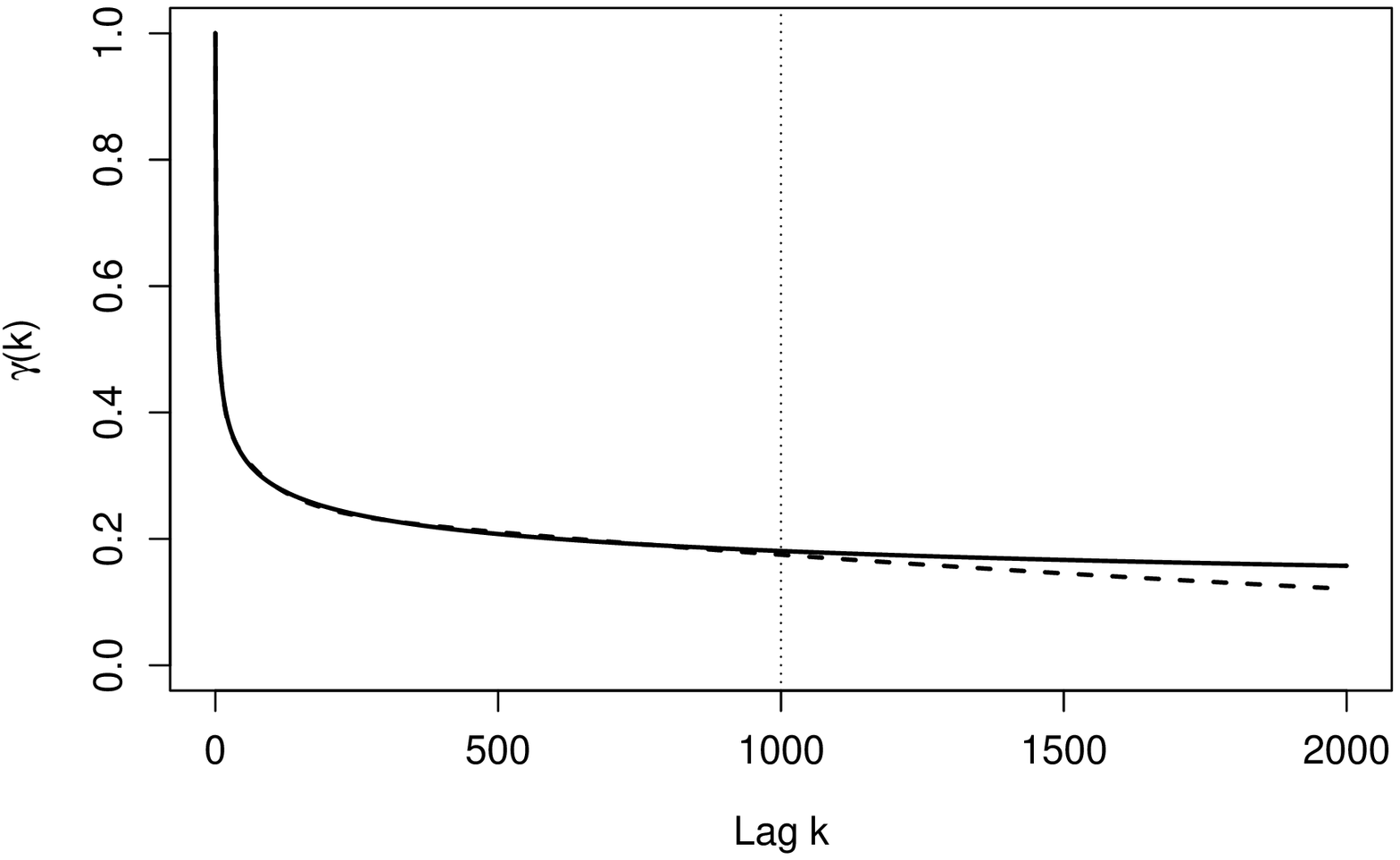}
      \caption{Left panel: The exact autocorrelation function (solid) when $H=0.9$, versus the autocorrelation function of the approximate model  (dashed), using $m=3$ and $k_{max}=1000$. Right panel: Similar using $m=4$. }
    \label{fig:acf}
    \end{center}
\end{figure}

We do notice that $m=4$ gives an almost
perfect match of the autocorrelation function up to
$k_{\max}$. For larger lags, the autocorrelation
function of the approximate fGn model will have an exponential decay,
hence we cannot match the hyperbolic decay of the exact
fGn. As a consequence, $k_{max}$ can be seen as a
trade-off between having a good fit for the first part of the
autocorrelation function versus tail behaviour. 

\begin{figure}[htb]
    \begin{center}
    \includegraphics[width=0.4\textwidth,angle=0]{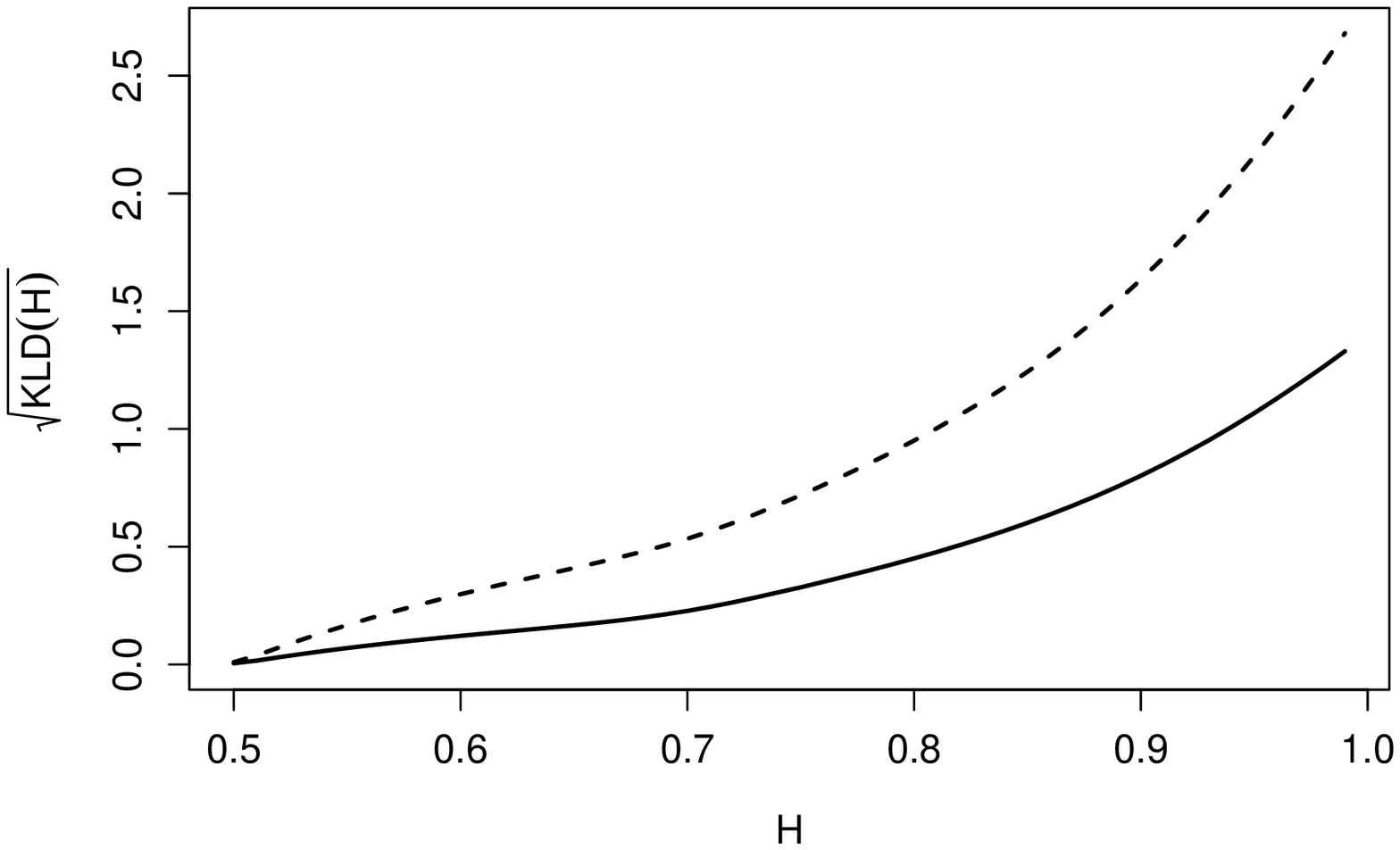}
        \hspace{0.5cm}
    \includegraphics[width=0.4\textwidth,angle=0]{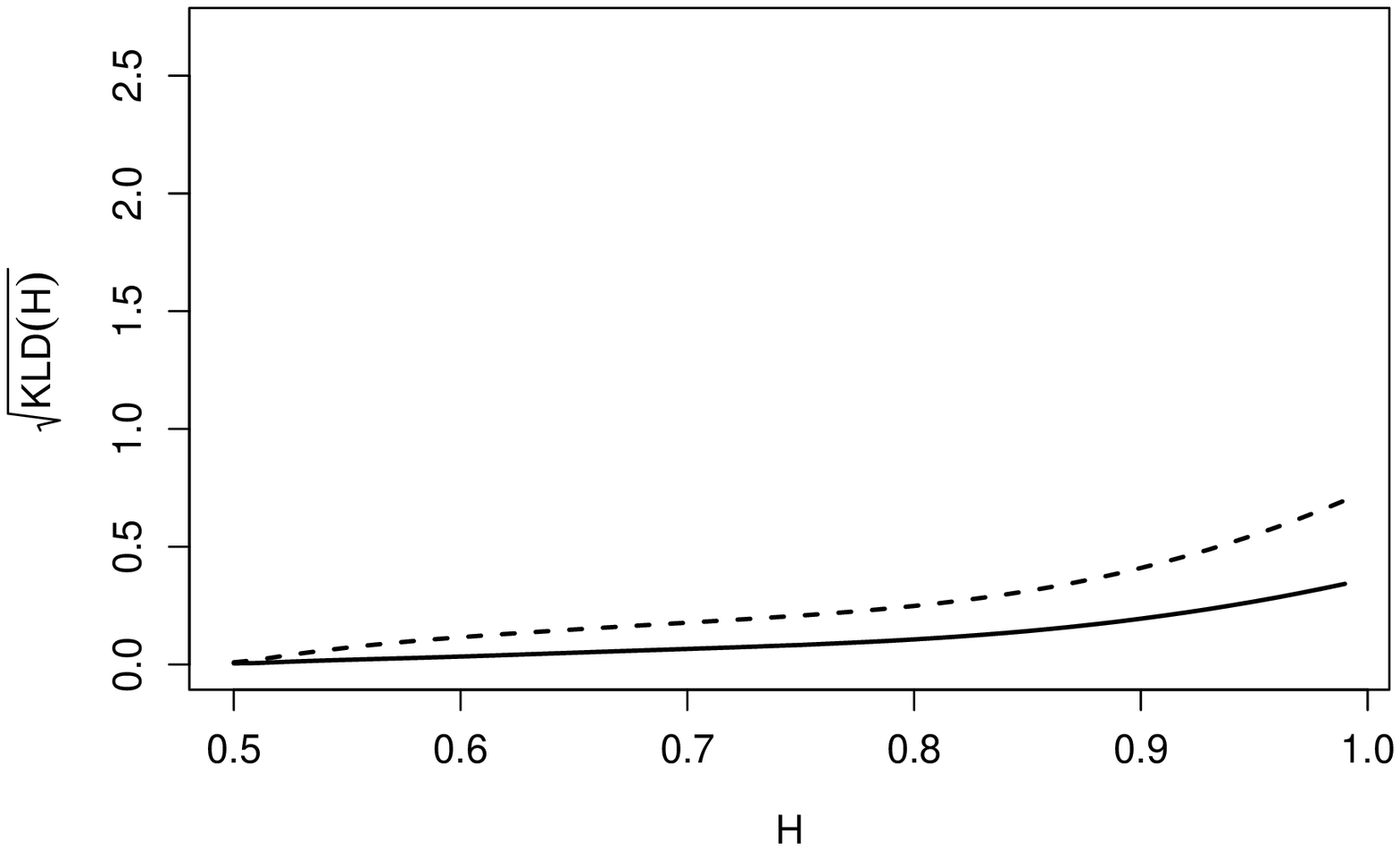}
    \caption{Left panel: The square-root of the Kullback-Leibler divergence as a
        function of $H$ for time series of length $n=500$ (solid)
        and $n=2000$ (dashed), using the approximate fGn model with $m=3$. 
       Right panel: Similar using 
        $m=4$. 
    }
    \label{fig:kld}
    \end{center}
\end{figure}

A different way to illustrate the difference between the approximate
and exact fGn models is in terms of the Kullback-Leibler divergence.
 This is a measure of complexity between probability
distributions, which here measures the information lost when
the approximate fGn model is used instead of the exact fGn model.
Figure~\ref{fig:kld} displays the square-root of the Kullback-Leibler
divergence for $n=500$ and $n=2000$, as a function of $H$. We notice that
$m=4$ clearly gives an improvement over $m=3$, in particular for larger
values of $H$. The loss in information when $n=2000$ compared to
$n=500$ is small, despite the fact that the autocorrelation function
is fitted only up to lag $k_{max}=1000$.

\section{Simulation results}\label{sec:results}

To evaluate the properties of the approximate fGn model, we now study the loss of accuracy when using the
approximate versus the exact fGn model, for estimation and prediction.
The results will demonstrate an impressive performance for
both the estimation and prediction exercises using the approximate fGn
model with $m=4$.

\subsection{Maximum likelihood estimation of $H$}

We first study the  loss of accuracy using the approximate versus the exact fGn model
 in maximum likelihood estimation of $H$. We fit the approximate model  using $m=3$ and $m=4$
 to simulated fGn series of
length $n=500$, with $N=1000$ replications. The error is evaluated in terms of the root mean squared error (RMSE)
and the mean absolute error (MAE) of $\tilde H_i-\hat H_i$, where
$\tilde H_i$ and $\hat H_i$ denote the estimates using the approximate
versus the exact fGn, for the $i$th replication. 

The results are
summarised in Table~\ref{tab:H} in which the true Hurst exponent ranges from 0.60 to 0.95. 
Using $m=3$, the Hurst
exponent is underestimated and the error is seen to increase with $H$, at least up to 0.90. 
The situation really improves
for $m=4$, in which  the error is small for all values of $H$. The
standard deviation estimates found from the empirical Fisher
information are more similar than the estimates themselves (results
not shown).

\begin{table}[ht!]
    \begin{center}
        \begin{tabular}{| l | rrr | rr | rr |}\hline
          & \multicolumn{3}{c|}{Average MLE of $H$} &\multicolumn{2}{c|}{$\mbox{RMSE}(\tilde H)$}  & \multicolumn{2}{c|}{$\mbox{MAE}(\tilde H)$} \\
          $ H $ &Exact &$ m=3$ & $m=4$ &  $m=3$ & $m=4$ & $m=3$ & $m=4$ \\\hline
          0.60 & 0.5998 & 0.5998 & 0.5998 & 0.0019 & 0.0007 & 0.0015 & 0.0006\\
           0.65 & 0.6481 & 0.6478 & 0.6480 & 0.0026 & 0.0008 & 0.0021
           & 0.0006\\
          0.70 & 0.7004 & 0.6997 & 0.7003 &   0.0033 &0.0008  &  0.0026  & 0.0006  \\
           0.75 & 0.7488 & 0.7472 & 0.7487 & 0.0032 & 0.0007 & 0.0025
           & 0.0006\\
          0.80 & 0.7998 & 0.7974 & 0.7996 &   0.0031 & 0.0006 &  0.0026 & 0.0005  \\
           0.85 & 0.8503 & 0.8471 & 0.8500 & 0.0035 & 0.0004 & 0.0032
           & 0.0004\\
          0.90 & 0.8999 & 0.8965 & 0.8997 &   0.0035 & 0.0003 &  0.0034  & 0.0003\\
           0.95 & 0.9500 & 0.9475 & 0.9499 & 0.0025 & 0.0002 & 0.0025
           & 0.0001
          \\\hline
        \end{tabular}
        \caption{The average of the maximum likelihood estimates of
            $H$, the root mean squared error and the absolute mean
            error using the exact versus the approximate models with $m=3$
            and $m=4$. The generated fGn processes are of length $n=500$ with $N=1000$ replications.}
        \label{tab:H}
    \end{center}
\end{table}
 
Figure~\ref{fig:mle} displays scatterplots of the  maximum likelihood 
estimates for the approximate model with $m=3$ and $4$, versus the estimates using the exact model,  
when $H=0.7$, $0.8$ and $0.9$.  The
inaccuracy for $m=3$ is clearly visible and increases with increasing values of 
$H$, while $m=4$ shows very good performance.  We have noticed that the same general remarks also
hold when we increase the length of the series to $n=2000$. 
The series then contain more information about $H$, and 
the error due to using $k_{max}=1000$ is negligible. In conclusion, we do get a very low loss of accuracy using
the approximate model  with $m=4$. This  is impressive, 
especially as it applies for all reasonable values of $H$ in the long memory range. 

\begin{figure}[ht]
    \begin{center}
    \includegraphics*[width=0.32\textwidth,angle=0]{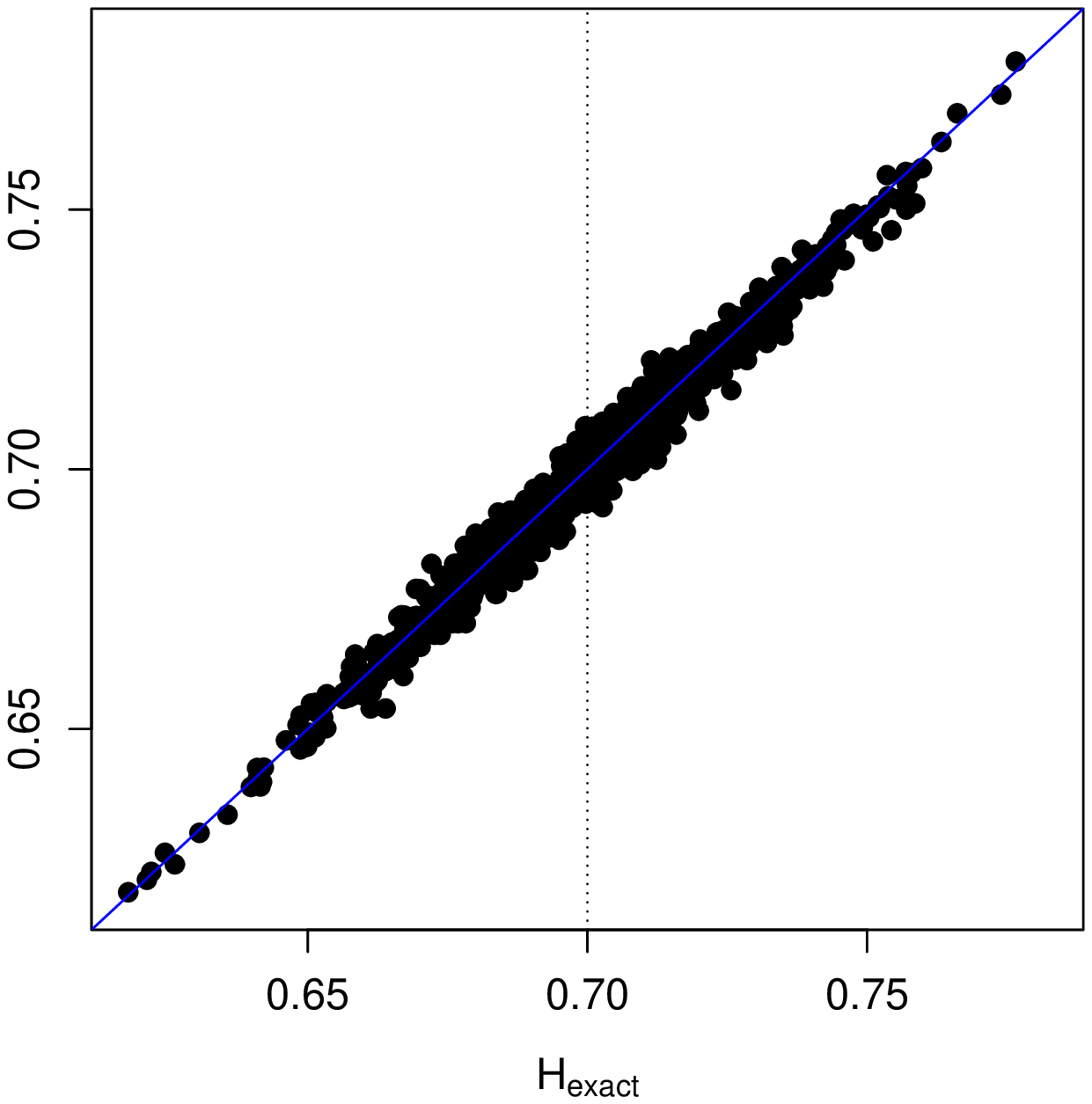}
    \includegraphics*[width=0.32\textwidth,angle=0]{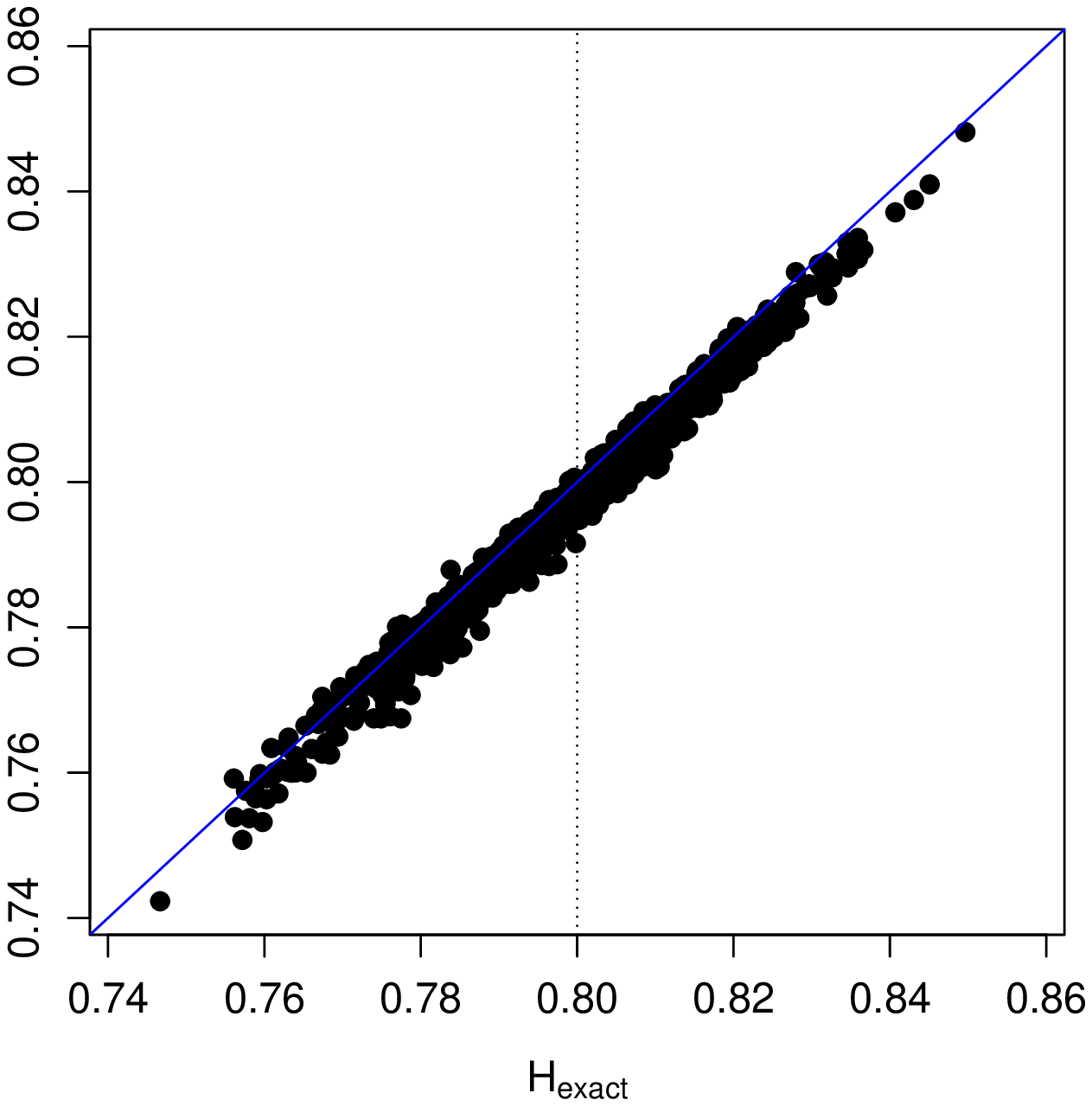} 
    \includegraphics*[width=0.32\textwidth,angle=0]{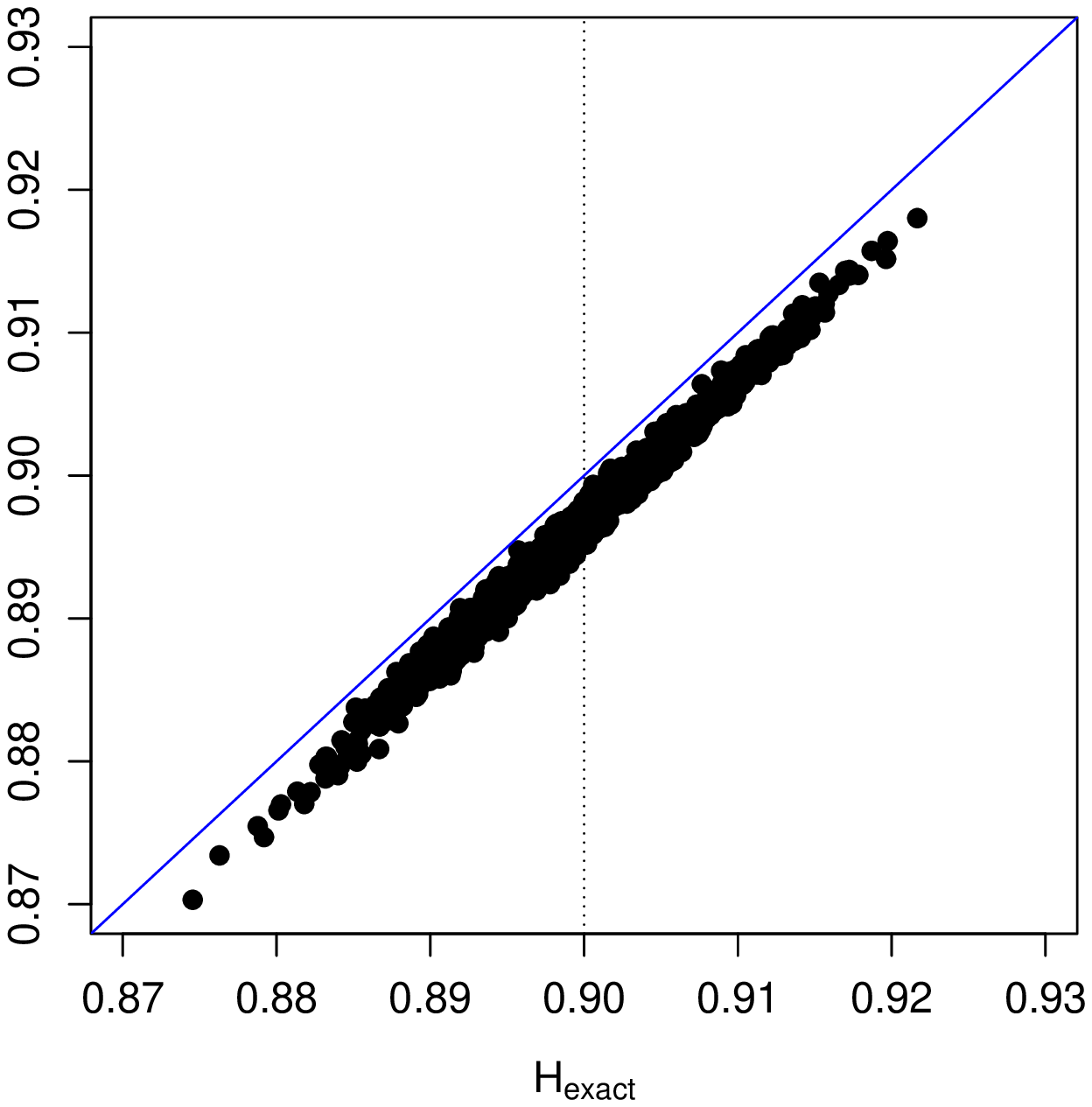}
    \includegraphics*[width=0.32\textwidth,angle=0]{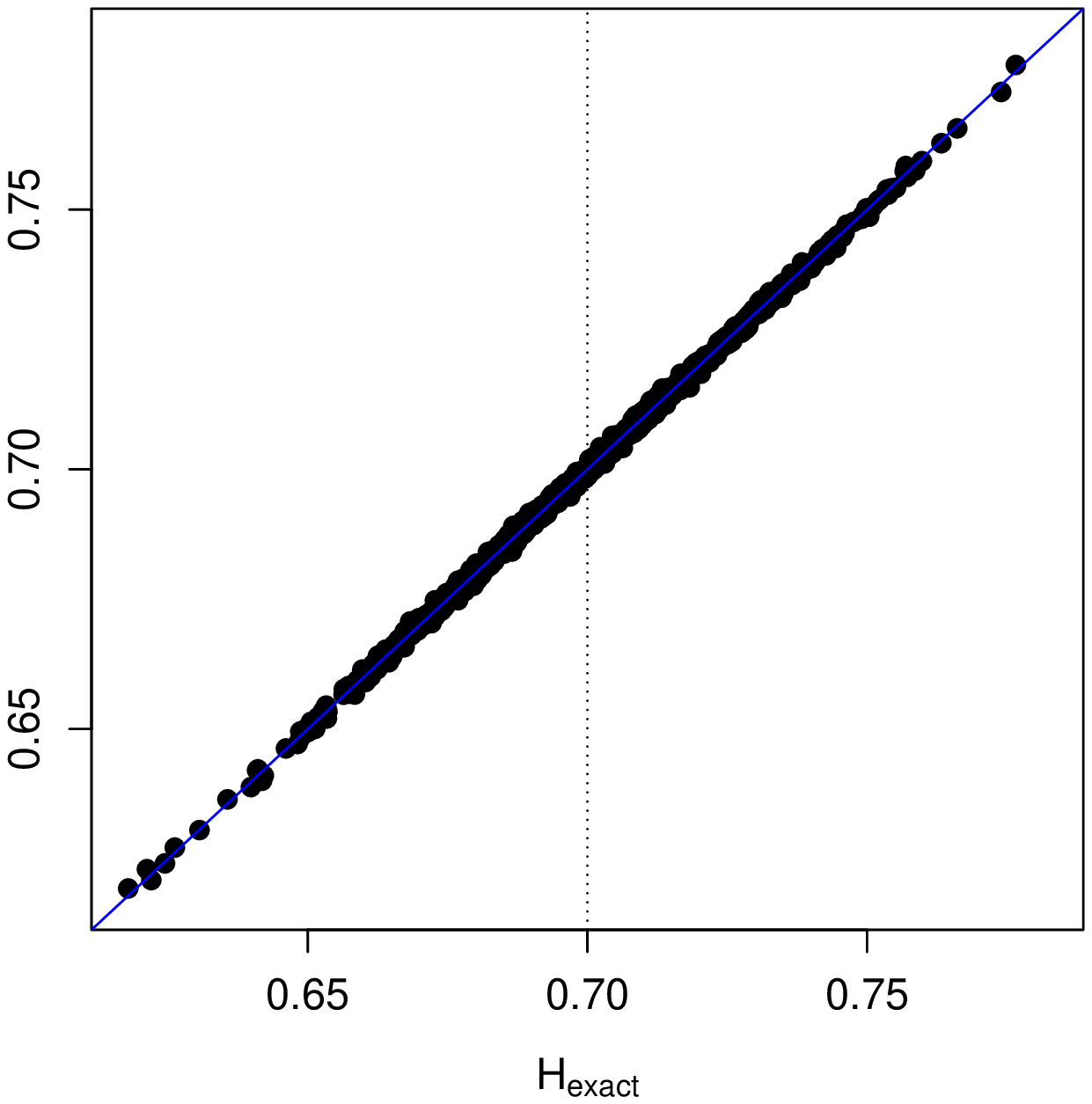}
    \includegraphics*[width=0.32\textwidth,angle=0]{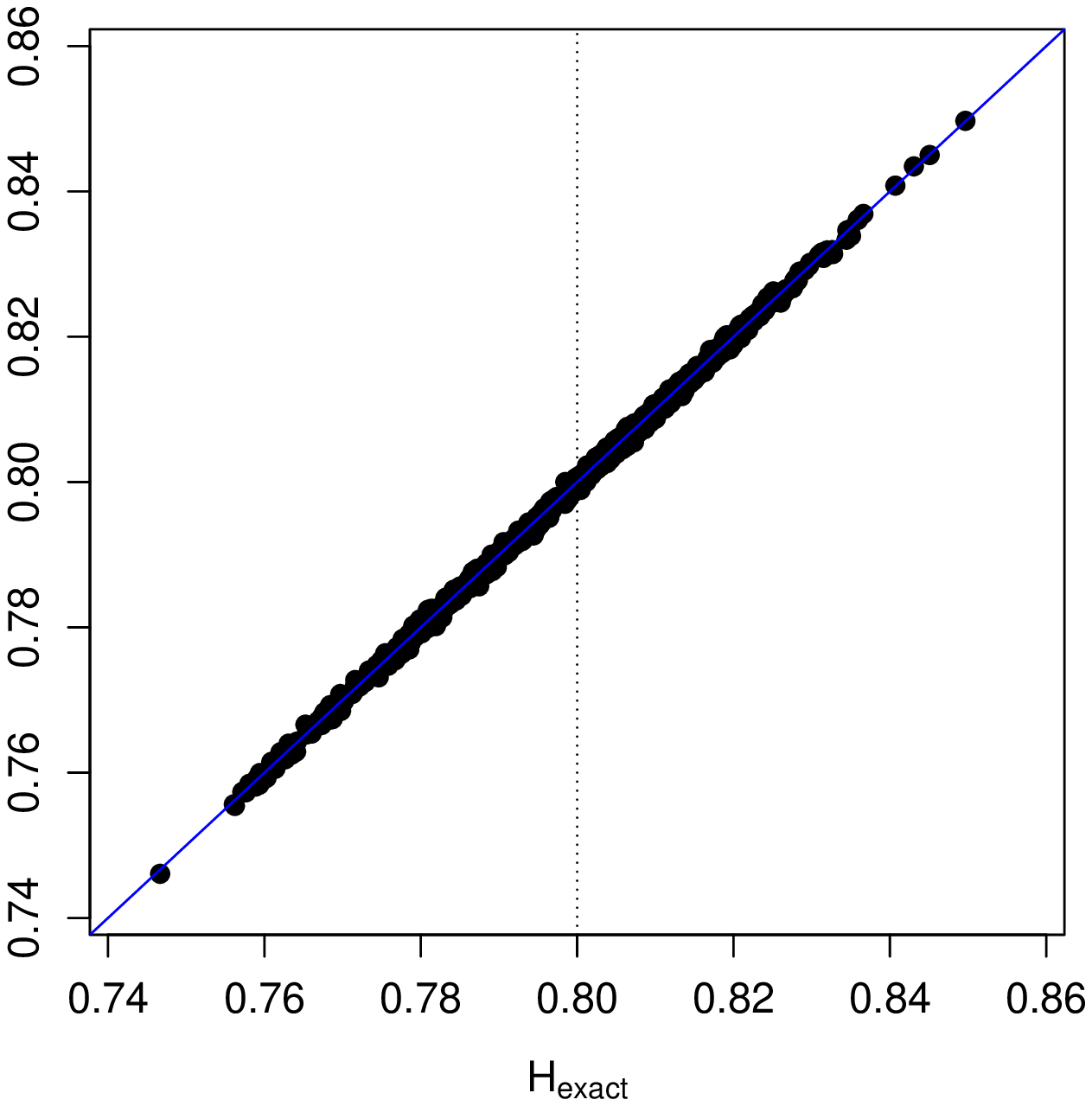}
    \includegraphics*[width=0.32\textwidth,angle=0]{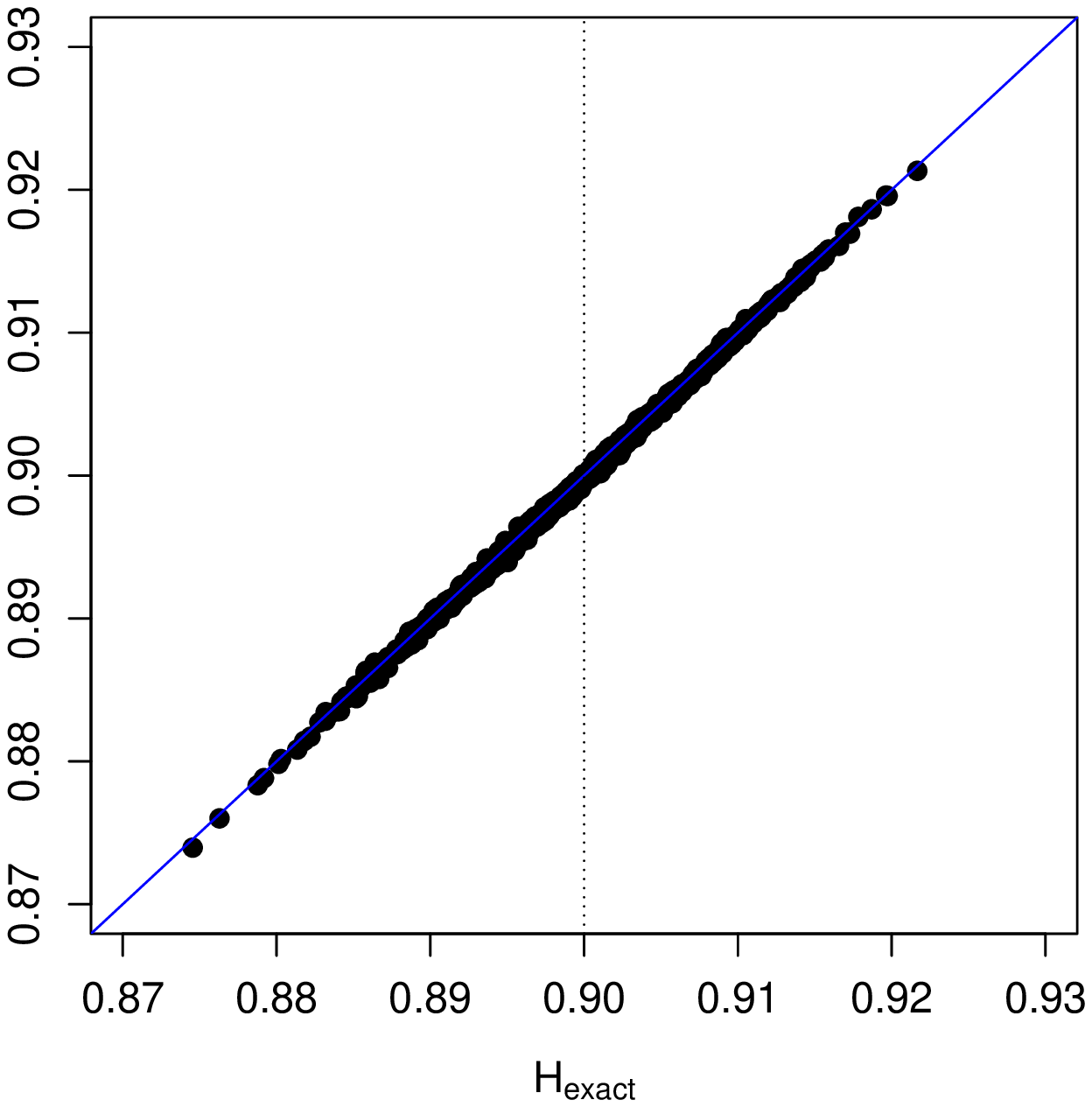}
    \caption{Upper panels: The maximum likelihood estimates of $H$ for  $N=1000$
        replications using the approximate
        fGn model with $m=3$ versus the exact fGn (labelled $H_{\mbox{\scriptsize{exact}}}$).   The true H-values are $H=0.7$ (left), $H=0.8$ (middle) and $H=0.9$ (right) 
        and the generated series have length $n=500$. 
        Lower panels: Similar using $m=4$.
       }
    \label{fig:mle}
    \end{center}
\end{figure}

\subsection{Predictive properties}

This section investigates the effect of the approximation error
when we observe an fGn process of length $n$ with fixed $H$, and then want to predict
future time points. The approximate model is implemented with $m=4$. To evaluate
the properties of the predictions, we consider the empirical mean
of the standardised absolute prediction error, 
\begin{equation}
    \text{err}_{\mu}(p) = \frac{1}{N}\sum_{i=1}^{N} 
    \frac{|\tilde{\mu}_{p,i} -{\mu}_{p,i}|}{\sigma_p},
    \label{eq:mean-pred}
\end{equation}
where $N$ is the number of replications. $\tilde{\mu}_{p,i}$ is the conditional expectation for $p$ time
points ahead from the $i$th replication using the approximate fGn
model. Correspondingly, ${\mu}_{p,i}$ is the conditional expectation using the
exact fGn model while $\sigma_p$ is the conditional standard deviation. 
To measure the error in the conditional standard deviation, we use
\begin{equation}
    \text{err}_{\sigma}(p) = \frac{\tilde{\sigma}_p}{\sigma_p} -1,
    \label{eq:sd-pred}
\end{equation}
which does not depend on the replication.

\begin{figure}[ht]
    \begin{center}
    \includegraphics[width=0.4\textwidth,angle=0]{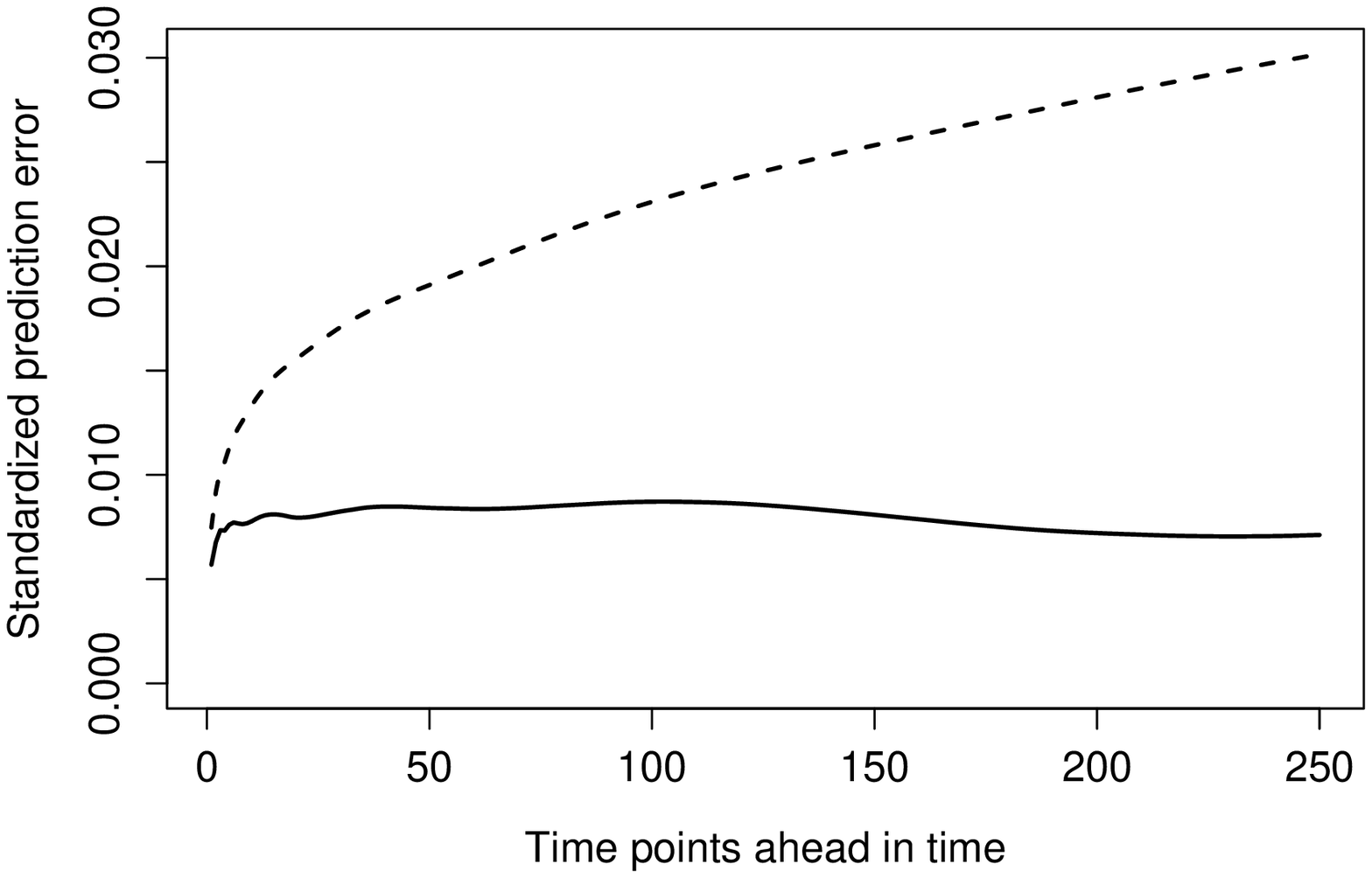}\hspace{0.5cm}
    \includegraphics[width=0.4\textwidth,angle=0]{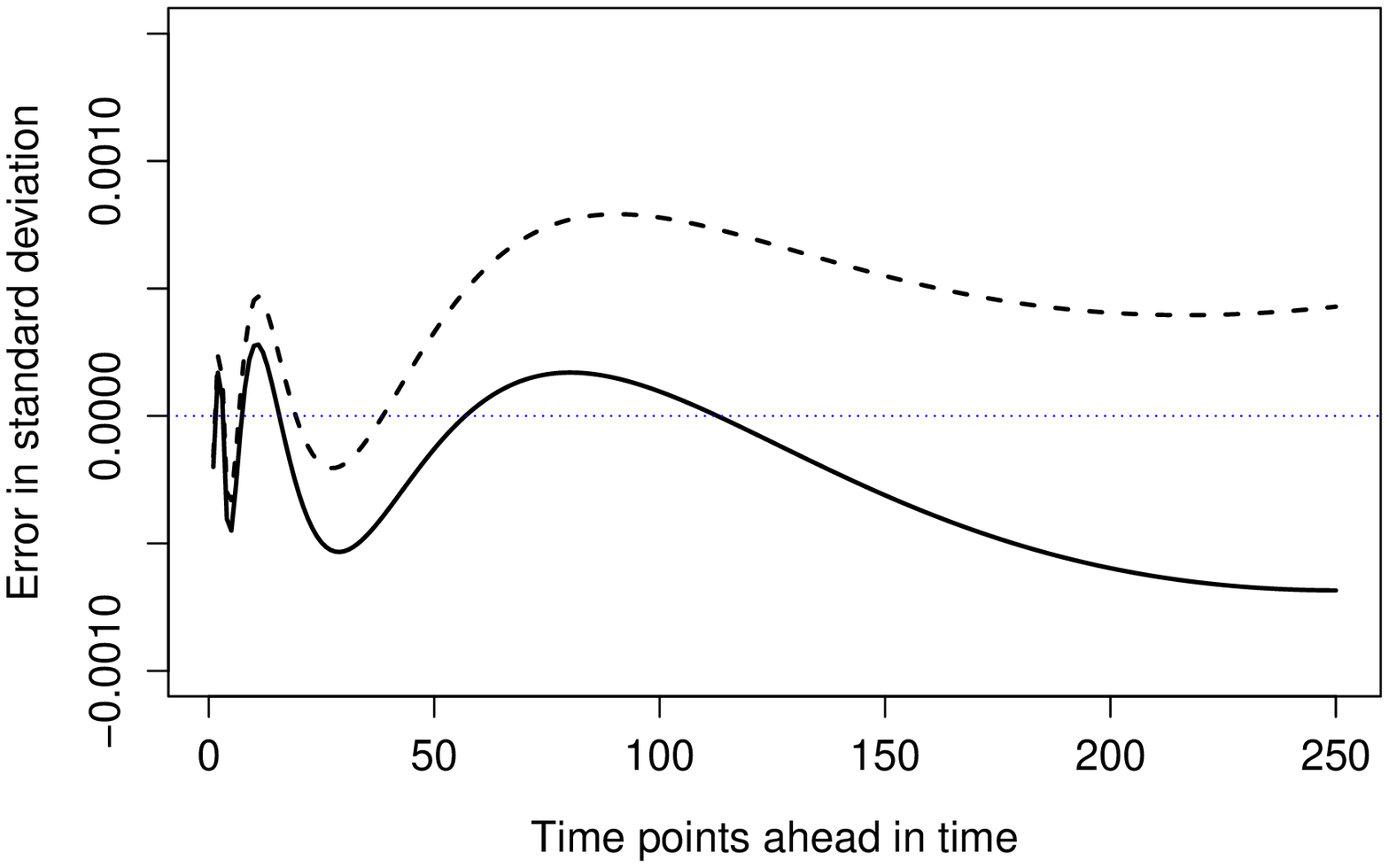}
    \caption{Left panel: The prediction error for the mean
        \eqref{eq:mean-pred}, predicting up to $250$ points ahead,  
        when $H=0.8$.
        The observed time series is either of length $n=500$
        (solid) or $n=2000$ (dashed).  Right panel: The
        similar error in the conditional standard deviation
        \eqref{eq:sd-pred}.
    }
    \label{fig:prederr}
    \end{center}
\end{figure}
 
The left panel of Figure~\ref{fig:prederr} illustrates the empirical
prediction error in \eqref{eq:mean-pred} for $p=1, \ldots, 250$
time-points ahead, following either $n=500$ or $n=2000$ observations.
The right panel shows the corresponding error in the prediction
standard deviation~\eqref{eq:sd-pred}. 
We only report results for $H=0.8$ as other values of the Hurst exponent give similar results.
We notice that the mean prediction error increases slightly when $n=2000$ compared
to $n=500$, which is explained by the increased error for lags
larger than $k_{\max}=1000$. Otherwise, both errors are relatively small and also quite stable
with $p$.

\section{Real data applications: Source separation and full Bayesian analysis}\label{sec:real-data}
This section demonstrates two different aspects of the approximate fGn model in real data applications.  
First, the approximate model can be used as a tool for source separation of a combined signal, 
for example representing underlying cycles or variations for different time scales.
This will be illustrated in analysing the Nile river dataset (available in \texttt{R} as \texttt{FGN::NileMin}). These data give 
annual water level minimas for the period 622 - 1284, measured at the Roda Nilometer near Cairo.  Second, the approximate model can easily be combined with other model components within the general framework of latent Gaussian models and fitted efficiently using
   \texttt{R-INLA}. This is demonstrated in analysing the Hadley Centre Central England Temperature series 
(HadCET), available at \href{http://www.metoffice.gov.uk/hadobs}{\texttt{http://www.metoffice.gov.uk/hadobs}}. These data  give mean monthly measurements of  surface air temperatures for Central England in the period 1659 - 2016. The two datasets are illustrated in  Figure~\ref{fig:real-data}.  

\begin{figure}[h]
    \begin{center}
             \includegraphics[width=0.4\textwidth]{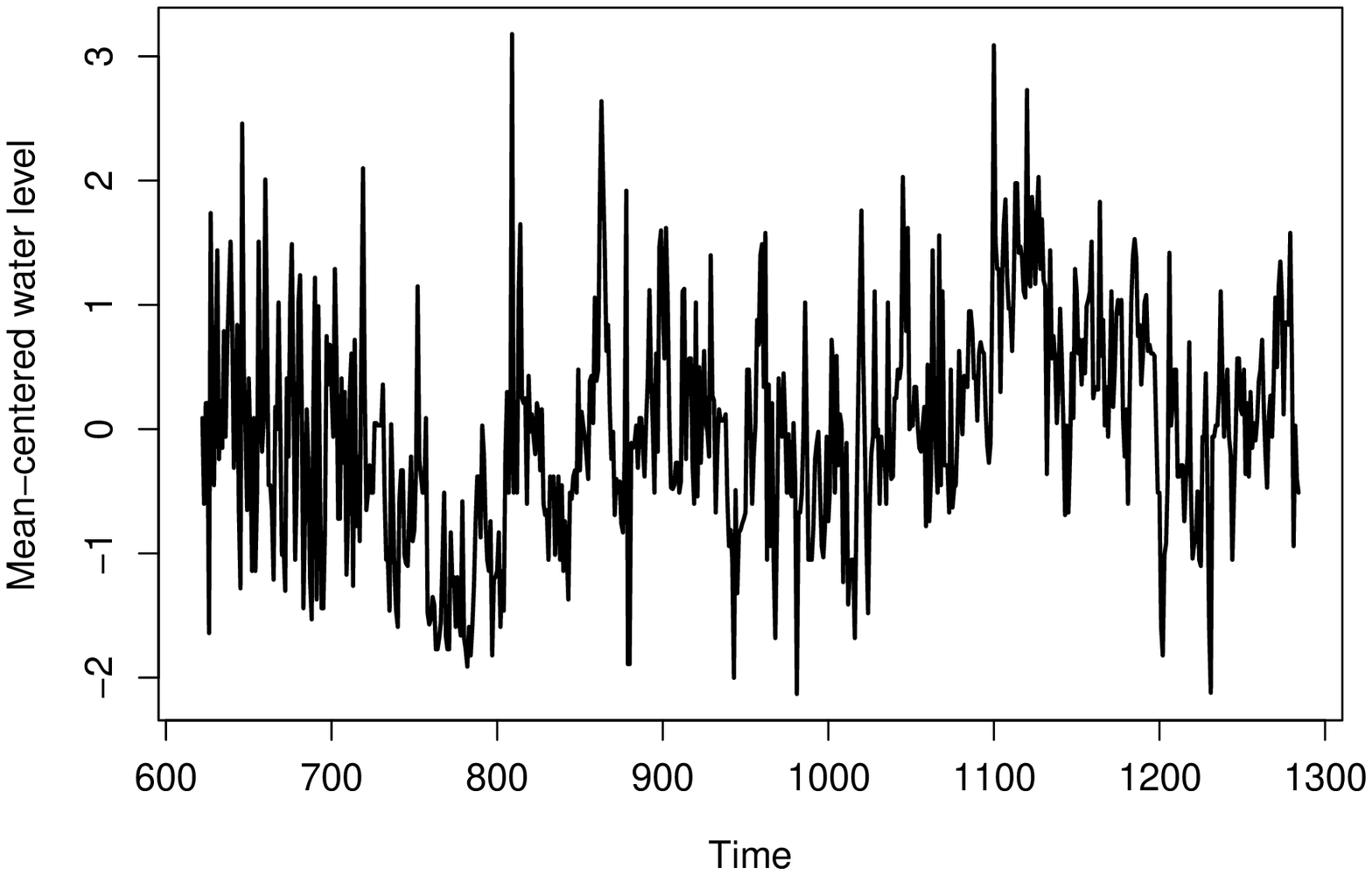}\hspace{0.5cm}
        \includegraphics[width=0.4\textwidth]{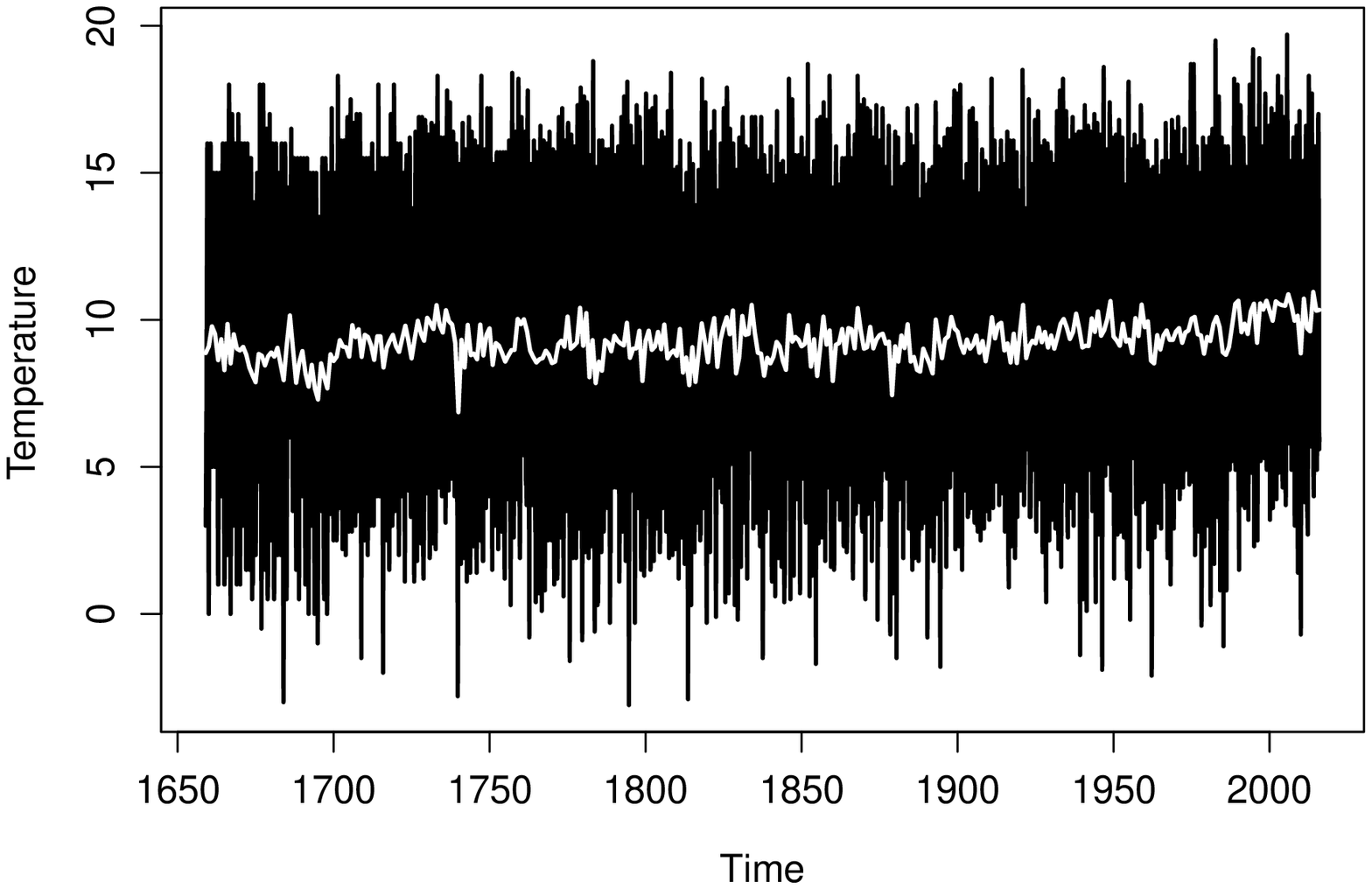}
         \caption{Left panel: Mean-centered annual minimum water level of the  Nile river. Right panel: Monthly mean surface air temperatures for Central England, also including the annual mean temperatures in white.}
        \label{fig:real-data}
    \end{center}
\end{figure}

\subsection{Signal separation for the Nile river annual minima}\label{sec:nile}
The Nile river dataset is a widely studied time series \citep{beranbok, eltahir:96} often used as an example of a real fGn process \citep{koutsoyiannis:02, benmehdi:11}.  Analysis of this dataset led to the discovery of the Hurst phenomenon \citep{hurst:51}. For hydrological time series, this phenomenon has  been explained as the tendency of having irregular clusters of wet and dry periods and can be related to characteristics of the fluctuations of the series at different temporal scales \citep{koutsoyiannis:02}.  

We can easily fit the exact fGn model to this dataset as the process is observed directly and the length of the series is only $n=663$. The maximum likelihood estimate for the Hurst exponent is $\hat H = 0.831$. Using the approximate fGn model with $m=4$, we get $\hat H=0.829$. The resulting four estimated weighted AR(1) components are illustrated in Figure~\ref{fig:nilenkomp}. The fitted autocorrelation coefficients for these components equal $\mm{\phi}=(0.999,0.982,0.847,0.291)$, while the weights are  $\mm{w}=(0.099, 0.129, 0.232, 0.540)$. The estimated standard deviation is $\hat \sigma=0.888$. 

As illustrated in Figure~\ref{fig:mapping}, the first autocorrelation coefficient will always be quite close to 1. This gives a slowly
varying trend, which in this case basically represents the mean.  
The second component also reflects a slowly varying signal, which can be interpreted to represent cycles of the water level fluctuations of about 200 - 250 years. The third component
seems to reflect shorter cycles of length 30 - 100 years. These cycles are seen to appear more irregularly and we also notice the 
tendency of having clusters of years with high and low water levels, respectively. The fourth component, which has the smallest autocorrelation coefficient and the largest weight, can be interpreted as weakly correlated annual noise.  These interpretations are  in correspondence with  the Hurst phenomenon, in which the components reflect signal fluctuations at different temporal scales. 

\begin{figure}[ht]
\begin{center}
\includegraphics*[width=0.4\textwidth,angle=0]{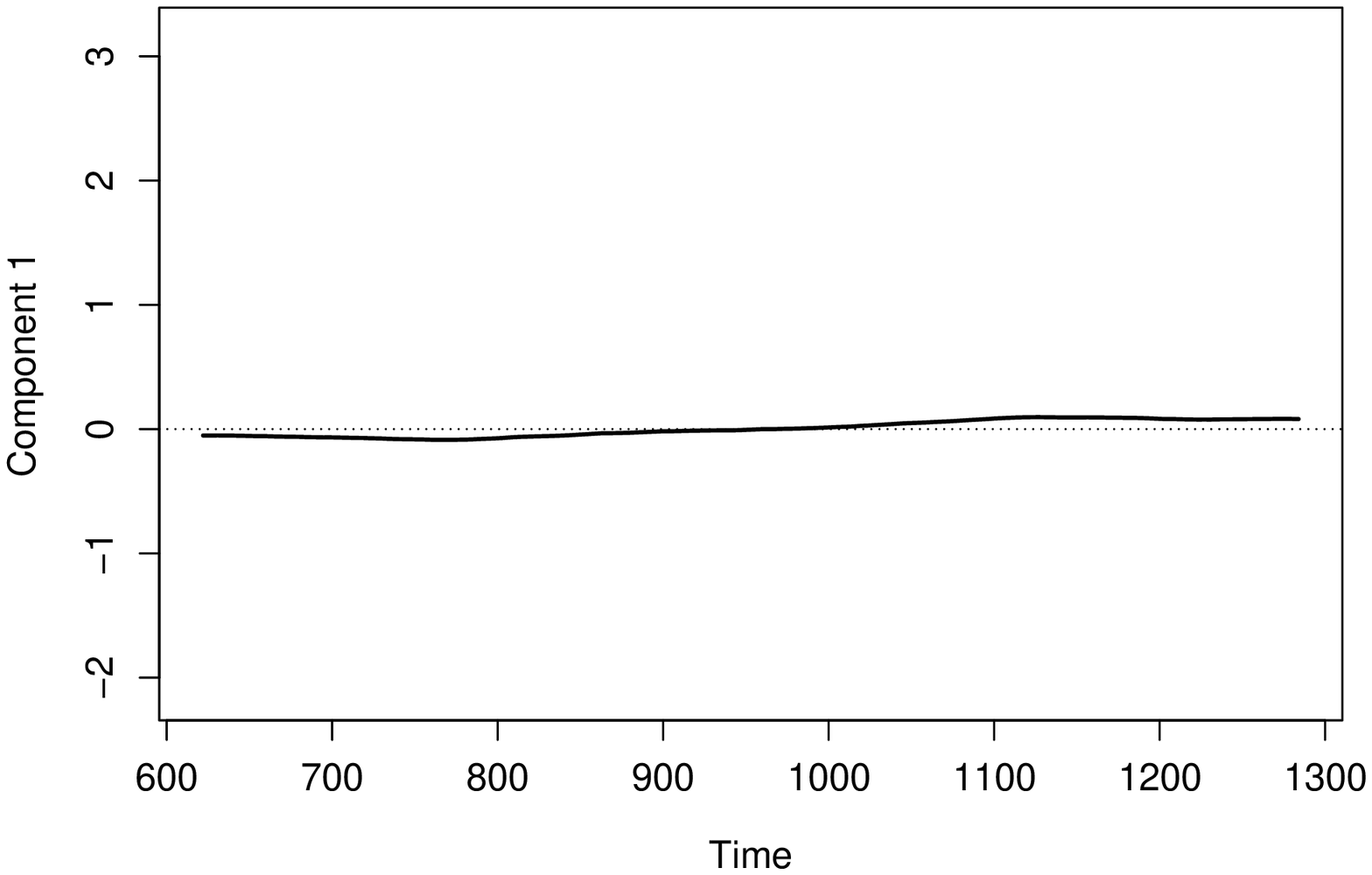}\hspace{0.5cm}
\includegraphics*[width=0.4\textwidth,angle=0]{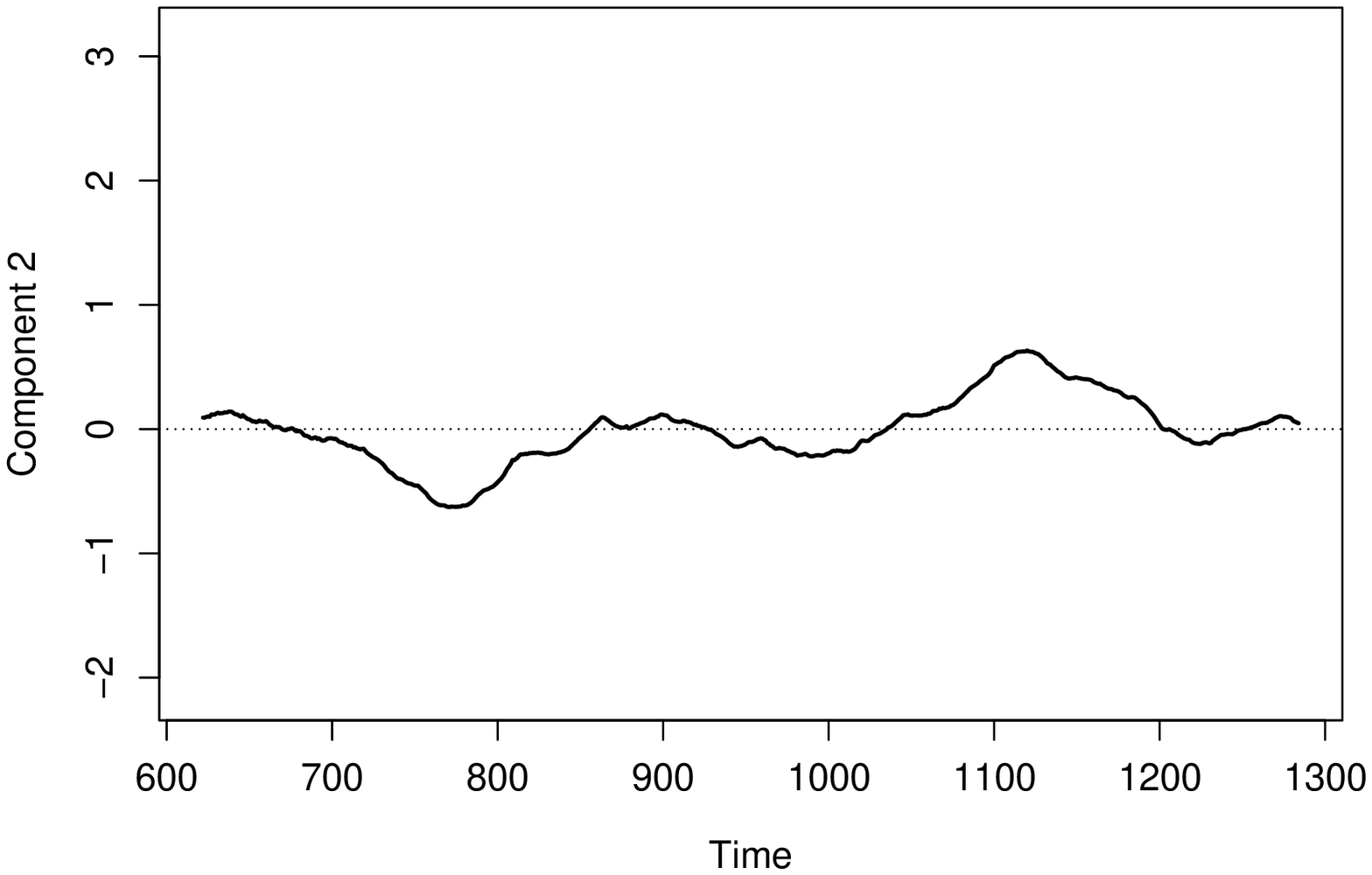} \\
\includegraphics*[width=0.4\textwidth,angle=0]{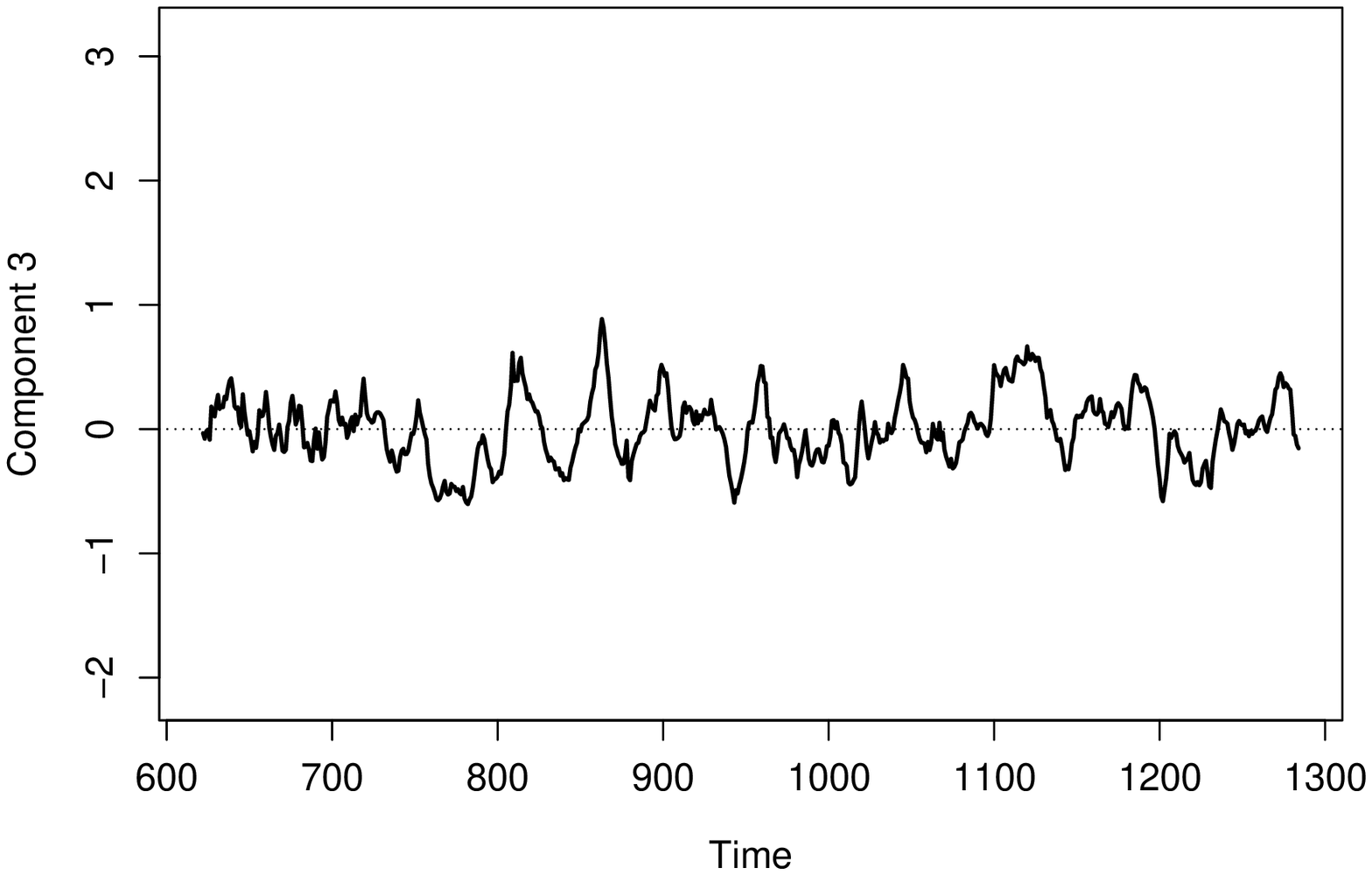}\hspace{0.5cm}
\includegraphics*[width=0.4\textwidth,angle=0]{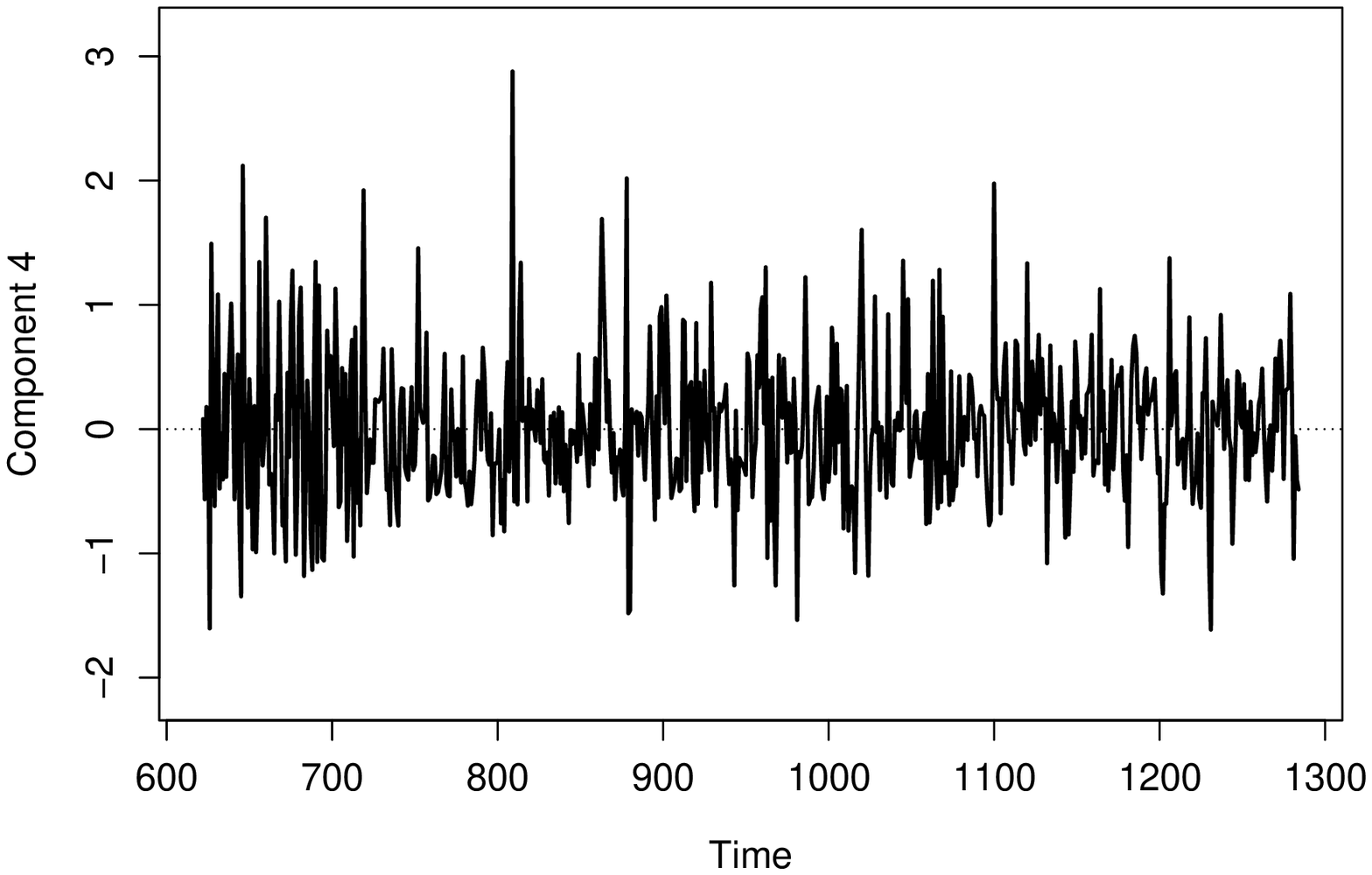}
\caption{The estimated weighted AR(1) components for mean-centered annual minimum water levels of the Nile river, using the approximate fGn model with $m=4$.}
\label{fig:nilenkomp}
\end{center}
\end{figure}

\subsection{Full Bayesian analysis of a temperature series }\label{sec:temp}

The HadCET series is the longest
existing instrumental record of monthly temperatures in the world. The observations started  in January 1659 and have been  updated monthly.  The observed temperatures do have uncertainties \citep{parker:05}, especially in the earliest years,  and has been revised several times \citep{manley:53, hadcet, parker:92}. We analyse temperatures  up to 
December 2016, which gives a total of $n=4296$
observations. 

In fitting a model to the given temperatures, we assume 
\begin{equation*}
    \E(y_t) = \beta_0 + \beta_1 t + s_t + x_t, \quad t=1,\ldots , 4296,
    \label{eq-hadcet:predictor}
\end{equation*}
where $y_t$ is the temperature in month $t$ (measured in degrees Celsius).  The given linear predictor  includes an intercept $\beta_0$,  a 
linear trend $\beta_1$, and  a seasonal effect $s_t$ of periodicity $q=12$
which captures monthly variations. This seasonal effect is modelled as an intrinsic GMRF of rank $n-q+1$, having precision parameter
$\tau_s$ \citep[p.~122]{ruebok} and scaled to have a generalized variance equal to 1 \citep{sorbye:13}. The term $x_t$ denotes the approximate fGn model with $m=4$, having precision parameter $\tau_x=\sigma^{-2}$. 

The parameters $\beta_0$ and $\beta_1$ are assigned vague Gaussian
priors, $\beta_i\sim N(0,10^3)$, while we use penalised 
complexity priors (PC priors) \citep{simpson:17} for all hyperparameters. 
This implies a type II Gumbel distribution
for the precision parameters $\tau_s$ and $\tau_x$, scaled 
using the probability
statement $P(\tau^{-1/2} > 1) = 0.01$.
The PC prior for $H$ \citep{sorbye:16} is scaled by assuming the
tail probability $P(H>0.9)=0.1$.

 \begin{figure}[h]
    \begin{center}
     \includegraphics[width=0.4\textwidth]{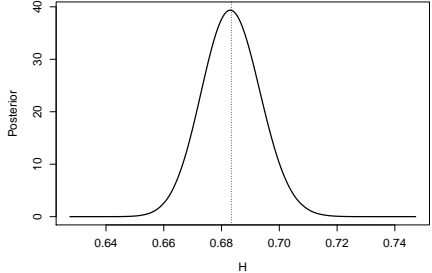}
           \caption{The marginal posterior $\pi(H\mid \mm{y})$, analysing the HadCET data.}
        \label{fig:hadcet-postmargs}
    \end{center}
\end{figure}

Analysis of the given model using an exact fGn term is infeasible in terms of
computational cost and memory usage. A MacBook Pro with 16GB of RAM
crashes due to memory shortage when analysing exact fGn processes of length $n>2500$. 
Using the approximate
fGn term with $m=4$, the full Bayesian analysis takes about 14 seconds.  The inference gives a significantly positive trend with posterior mean $\hat \beta_1 = 2.4\cdot 10^{-4}$  with $95\%$ credible interval $(1.2\cdot 10^{-4},3.7\cdot 10^{-4})$. This  corresponds to an overall increase in temperature of about $1.05 \pm 0.53$ degrees Celsius. 
The marginal posterior for $H$  is illustrated in Figure~\ref{fig:hadcet-postmargs}. The posterior mean is $\hat H = 0.684$ with  $95\%$ credible interval equal to $(0.664, 0.704)$. 

\section{Concluding remarks}\label{sec:conclusion}
In this paper, we obtain a remarkably 
accurate approximation of fGn 
using a weighted sum of
only four \AR1 components. Our approximate fGn model has a small loss of accuracy 
for the whole long memory range of $H$. The key idea to obtain this
is to ensure that the  approximate model captures the most essential part
 of the autocorrelation structure of the exact fGn model. This is achieved
 by appropriate weighting, matching the autocorrelation structure up to a specified maximum lag. 
 
By construction, the autocorrelation function of
the approximate model has an exponential decay for lags larger
than the specified maximum lag. This
implies that the approximate model does not satisfy formal definitions
of long memory processes. However, this trade-off is needed to make 
analysis of realistically complex models computationally feasible. The great benefit of
the resulting 
approximation is that it has a GMRF structure. This is crucial, especially 
as computations can be performed equally efficient in unconditional and conditional scenarios.  

An approximate model can never reflect the properties of the exact
model perfectly, but neither does a theoretical model in explaining an
observed data set. In theory, the fGn model corresponds to an
aggregation of an infinite number of \AR1 components which indicates
that the model is difficult to interpret in practice. The given decomposition 
of just a few \AR1 terms might provide a more realistic model. As an example, 
we have provided a decomposition of the Nile river data, which reflects 
fluctuations and cycles for different temporal scales. Such a decomposition could
also be valuable in analysing climatic time series. For example, long memory in  temperature series 
has been related  to an aggregation of a few simple underlying geophysical processes \citep{fredriksen:17}. 

Implementation of the approximate fGn model in \texttt{R-INLA} provides an easy-to-use tool to analyse
models with fGn structure. As demonstrated in the temperature example, we can 
easily combine the fGn model component with other terms in an additive linear predictor, for example covariates, other random effects
or deterministic effects like climate forcing.  Also, we do see a potential to incorporate fGn model components in  analysis of 
 spatial time series, for example by making use of the methodology in  \cite{lindgren:11}.  

\section*{Acknowledgement}
The authors acknowledge the Met Office Hadley Centre for making the HadCET data freely available at \href{http://www.metoffice.gov.uk/hadobs}{\texttt{http://www.metoffice.gov.uk/hadobs}}. We also  want to thank Finn Lindgren, Hege-Beate Fredriksen and Martin Rypdal for helpful discussions.
\bibliographystyle{apa}
\bibliography{bibliotek}{}

\begin{thebibliography}{}

\bibitem[\protect\astroncite{Baillie}{1996}]{baillie:96}
Baillie, R.~T. (1996).
\newblock Long memory processes and fractional integration in econometrics.
\newblock {\em Journal of Econometrics}, 73:5--59.

\bibitem[\protect\astroncite{Benhmehdi et~al.}{2011}]{benmehdi:11}
Benhmehdi, S., Makarava, N., Menhamidouche, N., and Holschneider, M. (2011).
\newblock Bayesian estimation of the self-similarity exponent of the {N}ile
  {R}iver fluctuation.
\newblock {\em Nonlinear Processes in Geophysics}, 18:441--446.

\bibitem[\protect\astroncite{Beran}{1994}]{beranbok}
Beran, J. (1994).
\newblock {\em Statistics for Long-Memory Processes}.
\newblock Chapman \& Hall/CRC, New York, 1st edition.

\bibitem[\protect\astroncite{Beran et~al.}{2013}]{beran:13}
Beran, J., Feng, Y., Ghosh, S., and Kulik, R. (2013).
\newblock {\em Long-Memory Processes - Probabilistic Properties and Statistical
  Methods}.
\newblock Springer, Heidelberg.

\bibitem[\protect\astroncite{Beran et~al.}{2010}]{beran:10}
Beran, J., Sch\"{u}tzner, M., and Ghosh, S. (2010).
\newblock From short to long memory: Aggregation and estimation.
\newblock {\em Computational Statistics and Data Analysis}, 54:2432--2442.

\bibitem[\protect\astroncite{Cont}{2005}]{cont:05}
Cont, R. (2005).
\newblock Long range dependence in financial markets.
\newblock In {\em Fractals in Engineering}, pages 159--180. Springer, London.

\bibitem[\protect\astroncite{Doukhan et~al.}{2003}]{taqqu:03}
Doukhan, P., Oppenheim, G., and Taqqu, M. (2003).
\newblock {\em Theory and Applications of Long-Range Dependence}.
\newblock Birkhauser Boston, c/o Springer-Verlag, New York Inc.

\bibitem[\protect\astroncite{Durbin}{1960}]{durbin1960}
Durbin, J. (1960).
\newblock The fitting of time-series models.
\newblock {\em Review of the International Statistical Institute}, 28:233--244.

\bibitem[\protect\astroncite{Eltahir}{1996}]{eltahir:96}
Eltahir, E. A.~B. (1996).
\newblock El {N}i{\~n}o and the natural variability in the flow of the {N}ile
  {R}iver.
\newblock {\em Water Resources Research}, 32:131--137.

\bibitem[\protect\astroncite{Franzke}{2012}]{franzke:12}
Franzke, C. (2012).
\newblock Nonlinear trends, long-range dependence, and climate noise properties
  of surface temperature.
\newblock {\em Journal of Climate}, 25:4172--4182.

\bibitem[\protect\astroncite{Fredriksen and Rypdal}{2017}]{fredriksen:17}
Fredriksen, H.-B. and Rypdal, M. (2017).
\newblock Long-range persistence in global surface temperatures explained by
  linear multibox energy balance models.
\newblock {\em Journal of Climate}, doi.org/10.1175/JCLI-D-16-0877.1.

\bibitem[\protect\astroncite{Golub and Loan}{1996}]{golub:96}
Golub, G. and Loan, C.~V. (1996).
\newblock {\em Matrix Computations}.
\newblock Johns Hopkins University Press, Baltimore, 3rd edition.

\bibitem[\protect\astroncite{Granger}{1980}]{granger1980}
Granger, C. W.~J. (1980).
\newblock Long memory relationships and the aggregation of dynamic models.
\newblock {\em Journal of Econometrics}, 14:227--238.

\bibitem[\protect\astroncite{Haldrup and Vald{\'e}s}{2017}]{aggregationsims}
Haldrup, N. and Vald{\'e}s, J. E.~V. (2017).
\newblock Long memory, fractional integration and cross-sectional aggregation.
\newblock {\em Journal of Econometrics}, 199:1--11.

\bibitem[\protect\astroncite{Hosking}{1984}]{hosking:84}
Hosking, J. R.~M. (1984).
\newblock Modeling persistence in hydrological time series using fractional
  differencing.
\newblock {\em Water Resources Research}, 20:1898--1908.

\bibitem[\protect\astroncite{Hurst}{1951}]{hurst:51}
Hurst, H.~E. (1951).
\newblock Long-term storage capacities of reservoirs.
\newblock {\em Transactions of the American Society for Civil Engineers},
  116:770--799.

\bibitem[\protect\astroncite{Koutsoyiannis}{2002}]{koutsoyiannis:02}
Koutsoyiannis, D. (2002).
\newblock The {H}urst phenomenon and fractional {G}aussian noise made easy.
\newblock {\em Hydrological Sciences}, 47:573--595.

\bibitem[\protect\astroncite{Levinson}{1947}]{levinson1947}
Levinson, N. (1947).
\newblock The {W}iener (root mean square) error criterion in filter design and
  prediction.
\newblock {\em Journal of Mathematics and Physics}, 25:261--278.

\bibitem[\protect\astroncite{Lindgren et~al.}{2011}]{lindgren:11}
Lindgren, F., Rue, H., and Lindstr\"{o}m, J. (2011).
\newblock An explicit link between {G}aussian fields and {G}aussian {M}arkov
  random fields: {T}he stochastic partial differential equation approach (with
  discussion).
\newblock {\em Journal of the Royal Statistical Society, Series B},
  73:423--498.

\bibitem[\protect\astroncite{Mandelbrot and Wallis}{1969}]{mandelbrot:69}
Mandelbrot, B.~B. and Wallis, J.~R. (1969).
\newblock Global dependence in geophysical records.
\newblock {\em Water Resources Research}, 5:321--340.

\bibitem[\protect\astroncite{Manley}{1953}]{manley:53}
Manley, G. (1953).
\newblock The mean temperature of central {E}ngland, 1698 to 1952.
\newblock {\em Quarterly Journal of the Royal Meteorological Society},
  79:242--261.

\bibitem[\protect\astroncite{Manley}{1974}]{hadcet}
Manley, G. (1974).
\newblock Central {E}ngland temperatures: {M}onthly means 1659 to 1973.
\newblock {\em Quarterly Journal of the Royal Meteorological Society},
  100:389--405.

\bibitem[\protect\astroncite{McLeod and Hipel}{1978}]{mcleod:78}
McLeod, A.~I. and Hipel, K.~W. (1978).
\newblock Preservation of the rescaled adjusted range: 1. {A} reassessment of
  the {H}urst phenomenon.
\newblock {\em Water Resources Research}, 14:491--508.

\bibitem[\protect\astroncite{McLeod et~al.}{2007}]{mcleod:07}
McLeod, A.~I., Yu, H., and Krougly, Z.~L. (2007).
\newblock Algorithms for linear time series analysis: With {R} {P}ackage.
\newblock {\em Journal of Statistical Software}, 23:1--26.

\bibitem[\protect\astroncite{Parker and Horton}{2005}]{parker:05}
Parker, D.~E. and Horton, E.~B. (2005).
\newblock Uncertainties in the central {E}ngland temperature series since 1878
  and some changes to the maximum and minimum series.
\newblock {\em International Journal of Climatology}, 25:1173--1188.

\bibitem[\protect\astroncite{Parker et~al.}{1992}]{parker:92}
Parker, D.~E., Legg, T.~P., and Folland, C.~K. (1992).
\newblock A new daily central {E}ngland temperature series, 1772-1991.
\newblock {\em International Journal of Climatology}, 12:317--342.

\bibitem[\protect\astroncite{Rue}{2001}]{rue:01}
Rue, H. (2001).
\newblock Fast sampling of {G}aussian {M}arkov random fields.
\newblock {\em Journal of the Royal Statistical Society, Series B},
  63:325--338.

\bibitem[\protect\astroncite{Rue and Held}{2005}]{ruebok}
Rue, H. and Held, L. (2005).
\newblock {\em Gaussian Markov Random Fields: Theory and Applications}.
\newblock Chapman \& Hall/CRC, London.

\bibitem[\protect\astroncite{Rue et~al.}{2009}]{inlartikkel}
Rue, H., Martino, S., and Chopin, N. (2009).
\newblock Approximate {B}ayesian inference for latent {G}aussian models using
  integrated nested {L}aplace approximations (with discussion).
\newblock {\em Journal of the Royal Statistical Society, Series B},
  71:319--392.

\bibitem[\protect\astroncite{Rypdal and Rypdal}{2014}]{rypdal:14}
Rypdal, M. and Rypdal, K. (2014).
\newblock {L}ong-memory effects in linear response models of {E}arth's
  temperature and implications for future global warming.
\newblock {\em Journal of Climate}, 27:5240--5258.

\bibitem[\protect\astroncite{Simpson et~al.}{2017}]{simpson:17}
Simpson, D., Rue, H., Riebler, A., Martins, T.~G., and S{\o}rbye, S.~H. (2017).
\newblock Penalising model component complexity: A principled, practical
  approach to constructing priors.
\newblock {\em Statistical Science}, 232:1--28.

\bibitem[\protect\astroncite{S{\o}rbye and Rue}{2014}]{sorbye:13}
S{\o}rbye, S.~H. and Rue, H. (2014).
\newblock Scaling intrinsic {G}aussian {M}arkov random field priors in spatial
  modelling.
\newblock {\em Spatial Statistics}, 8:39--51.

\bibitem[\protect\astroncite{S{\o}rbye and Rue}{2017}]{sorbye:16}
S{\o}rbye, S.~H. and Rue, H. (2017).
\newblock Fractional {Gaussian} noise: Prior specification and model
  comparison.
\newblock {\em Environmetrics}, doi:10.1002/env.2457.

\bibitem[\protect\astroncite{Toussoun}{1925}]{toussoun:1925}
Toussoun, O. (1925).
\newblock M{'e}moire sur l'histoire du {N}il.
\newblock In {\em M{'e}moires a l'Institut d'Egypte}, volume~18, pages
  366--404.

\bibitem[\protect\astroncite{Trench}{1964}]{trench:64}
Trench, W.~F. (1964).
\newblock An algorithm for the inversion of finite {T}oeplitz matrices.
\newblock {\em Journal of the Society for Industrial and Applied Mathematics},
  12:515--522.

\bibitem[\protect\astroncite{Willinger et~al.}{1996}]{willinger:96}
Willinger, W., Paxson, V., and Taqqu, M.~S. (1996).
\newblock Self-similarity and heavy tails: Structural modeling of network
  traffic.
\newblock In {\em A Practical Guide to Heavy Tails: Statistical Techniques and
  Applications}, pages 27--53. Birkhauser, Boston.

\end{thebibliography}
\end{document}